\documentclass[journal,final]{IEEEtran}
\usepackage{epsfig,rotating,setspace,latexsym,amsmath,epsf,amssymb,bm}
\usepackage{color}

\usepackage{amsmath}
\usepackage{rotating,latexsym,amsmath,amssymb,bm}
\usepackage{cite}
\usepackage{amsthm}
\usepackage{subfigure}
\usepackage{graphicx}
\def\nb0{{\mathbf{0}}}
\def\nb1{{\mathbf{1}}}

\def\ncalA{{\mathcal{A}}}
\def\ncalB{{\mathcal{B}}}
\def\ncalC{{\mathcal{C}}}

\def\ncalH{{\mathcal{H}}}

\def\ncalO{{\mathcal{O}}}

\def\ncalS{{\mathcal{S}}}

\def\ncalW{{\mathcal{W}}}

\def\ncalY{{\mathcal{Y}}}

\def\nbbN{{\mathbb{N}}}

\def\nbbR{{\mathbb{R}}}

\newtheorem{theorem}{Theorem}
\newtheorem{example}{Example}
\def\E{\mathbb{E}}

\usepackage{bbm}
\usepackage{url}

\usepackage{wrapfig}
\usepackage{dsfont}
\usepackage{wasysym}
\usepackage{hyperref}
\usepackage{balance}
\usepackage[ruled,vlined]{algorithm2e}

\theoremstyle{plain}

\newtheorem{prop}{Proposition}

\theoremstyle{definition}
\newtheorem{Def}{Definition}
\newtheorem{assum}{Assumption}

\usepackage{thmtools}
\declaretheoremstyle[
  spaceabove=\topsep, spacebelow=\topsep,
  headfont=\normalfont\bfseries,
  notefont=\mdseries, notebraces={(}{)},
  bodyfont=\normalfont,
  postheadspace=1em,
  qed=\qedsymbol
]{mythmstyle}
\declaretheoremstyle[
  spaceabove=\topsep, spacebelow=\topsep,
  headfont=\normalfont\bfseries,
  notefont=\mdseries, notebraces={(}{)},
  bodyfont=\normalfont,
  postheadspace=1em,
  qed=$\diamond$
]{mythmstyle}
\declaretheorem[style=mythmstyle]{remark}

\DeclareMathOperator{\psaa}{\pi_{\text{SAA}}}
\DeclareMathOperator{\pts}{\pi_{\text{TS}}}

\begin{document}

\allowdisplaybreaks

\sloppy

\title{On the Value of Online Learning \\ for Radar Waveform Selection}

\author{Charles E. Thornton, \IEEEmembership{Member, IEEE}, and R. Michael Buehrer, \IEEEmembership{Fellow, IEEE}
\thanks{A preliminary version of this work \cite{thornton2022cognitive} was presented at the 2023 IEEE Radar Conference, San Antonio, TX. C.E. Thornton is with the Virginia Tech National Security Institute, Virginia Tech, Blacksburg, VA, USA, 24061. R.M. Buehrer is with Wireless $@$ Virginia Tech, Department of ECE, Virginia Tech, Blacksburg, VA, USA, 24061. Correspondence: $thorntonc@vt.edu$.}}

\maketitle
\thispagestyle{plain}
\pagestyle{plain}
\vspace{-1cm}
\begin{abstract}
This paper attempts to characterize the kinds of physical scenarios in which an online learning-based cognitive radar is expected to reliably outperform a fixed rule-based waveform selection strategy, as well as the converse. We seek general insights through an examination of two decision-making scenarios, namely dynamic spectrum access and multiple-target tracking. The radar scene is characterized by inducing a state-space model and examining the structure of its underlying Markov state transition matrix, in terms of entropy rate and diagonality. It is found that entropy rate is a strong predictor of online learning-based waveform selection performance, while diagonality is a better predictor of fixed rule-based waveform selection performance. We show that these measures can be used to predict first and second-order stochastic dominance relationships, which can allow system designers to make use of simple decision rules instead of more cumbersome learning approaches under certain conditions. We validate our findings through numerical results for each application and provide guidelines for future implementations.

%In what kinds of physical scenarios is an online learning-based cognitive radar be expected to significantly outperform a fixed rule-based waveform selection strategy? We seek insight towards this fundamental question by applying tools from stochastic dominance, which provides a partial ordering over random outcomes. We examine the performance of an online learning-based approach relative to a fixed baseline for two scenarios of interest. First, we examine a dynamic spectrum access scenario and show that the online learning algorithm dominates a sense-and-avoid strategy in the long run. Then, we examine a much more general multi-target tracking problem, and find instances of the scenario where the cognitive radar is expected to most outperform an adaptive baseline. We conclude by presenting general insights about the value of online learning and provide guidelines for future implementations.
\end{abstract}

\begin{IEEEkeywords}
cognitive radar, online learning, radar spectrum access, target tracking
\end{IEEEkeywords}

\section{Introduction}
In many practical applications, a radar system's ability to detect, track, or classify targets of interest may be significantly improved by dynamically tailoring the transmitted waveform to complement the state of the physical scene \cite{sira2009waveform,calderbank2009waveform,kershaw1994optimal,gini2012waveform,blunt2016overview}. For example, dynamic waveform selection is particularly relevant for radar scenes which are rapidly varying, heavily cluttered or include dynamic interference from neighboring systems \cite{Sira2007,martone2021closing,kovarskiy2020evaluation}. Such scenarios are notably common in emerging vehicular radar applications \cite{feng2018cognitive,zamiri2022bayesian}. Additionally, if the radar application requires a high resolution, and targets can no longer be well-approximated by point scatterers, waveform diversity can provide substantial benefits in terms of target detection and classification \cite{rihaczek1996principles,bell1993information,Sowelam2000}. Such applications include UAV recognition \cite{wang2021deep}, `vision' for autonomous systems \cite{kim2015firefighting}, and space situational awareness \cite{ender2011radar}.

For any waveform adaptation process to be effective, one must develop a meaningful criterion, or \emph{measure of effectiveness}, by which waveforms may be selected \cite{hero2007foundations,LaScala2004intl,mitchell2018cost}. This presents a significant design challenge, since radar system designers often wish to optimize explicit \emph{task-based} performance measures such as the expected tracking error, probability of target detection, or expected target classification performance \cite{charlish2020implementing}. Furthermore, a reasonable measure of effectiveness must be readily computable on-the-fly for the entire catalog of candidate waveforms.

Unfortunately, these ideal task-based performance measures are frequently non-trivial, or impossible, to compute on-the-fly. Several important works aim to directly quantify the effect of the waveform on tracking performance \cite{kershaw1994optimal,kershaw1997waveform,Sira2007,nguyen2015adaptive}. However, this process becomes cumbersome beyond simple cases. The difficulty is accentuated when the radar scene is rapidly time-varying and prior information is limited, which is the case for practical target tracking problems. Thus, it is often more appropriate to use surrogate measures of effectiveness, such as the estimated SINR/SCNR, information-theoretic quantities, or normalized innovation squared of the tracking filter, which may be directly computed using observed radar measurements. The waveform selection process is then concerned with selecting the waveform which will optimize the surrogate measure of interest using statistical prediction techniques, human-designed rules based on prior design experience, or some combination thereof, given limited feedback from noisy measurements. 

%Much research has attempted to define performance criteria for cognitive radar \cite{charlish2015anticipation,charlish2020implementing}. Other important works include \cite{krishnamurthy2016partially,krishnamurthy2009optimal,doly2022waveform}, which primarily deal with POMDP formulation.

Due to the aforementioned practical challenges, a simple strategy with immediate appeal is to apply preordained, fixed, decision rules. Under this paradigm, certain waveforms are transmitted under channel states which are expected to be complementary. For example, a radar might select waveforms depending on the estimated target position using prior knowledge about the waveforms' ambiguity functions \cite{rihaczek1971radar}. However, such a scheme is limited for two reasons. First, it may be difficult to select decision rules which are effective for all scenarios of interest. Secondly, without extensive prior knowledge regarding quantities as diverse as the target motion model, measurement probabilities, and clutter distribution, it may be difficult to select suitable decision rules, especially if the set of candidate waveforms is very large.

A more nuanced strategy is to frame the waveform selection process as a Partially Observed Markov Decision Process (POMDP) and find an optimal or nearly-optimal waveform selection policy using stochastic optimization techniques \cite{charlish2020implementing,charlish2015anticipation,selvi2020reinforcement,thornton2020deep,ahmad2009optimality}. Indeed, the theory of POMDPs is well-developed, the framework naturally fits the radar tracking problem, and the existence of optimal solutions can be guaranteed. However, POMDPs are in general numerically intractable, and exact solutions require precise knowledge of the environment's time-varying behavior \cite{hero2007foundations}. Several important works have focused on applying structural properties to simplify the POMDP problem \cite{krishnamurthy2009optimal,krishnamurthy2016partially,krishnamurthy2002algorithms}. However, even when the problem is simplified, the computational resources required to solve the POMDP can still be burdensome for realistic problem spaces, and the optimal policy must be computed offline. Further, the POMDP model assumes the radar's actions influence the evolution of future states, which is often not the case for practical waveform selection problems.

Given the practical challenges associated with the POMDP formulation, recent research in cognitive radar has focused on computational techniques to gradually infer, or learn, aspects of the scene in real-time, to dynamically select waveforms with very few restrictive assumptions. In particular, online learning presents a compelling framework for adaptive waveform selection, since the problem of sequential decision making under uncertainty is handled explicitly using well-established mathematics. Perhaps more importantly, practical algorithms with impressive empirical performance and strong theoretical guarantees can be leveraged to drive the decision-making process \cite{thornton2021constrained}. The online learning approach for adaptive waveform selection has been studied in several previous works \cite{thornton2021constrained,thornton2022online,howard2022distributed}, and is shown to be an effective strategy in many scenarios. 

However, the crucial question of whether such learning approaches will \emph{reliably} outperform simpler, fixed rule-based strategies over a broad variety of physical conditions remains understudied in the open literature. This is the question which we pursue here. We acknowledge that a practical characterization of online learning performance for all waveform selection tasks is beyond the scope of a single investigation, but present insight towards the general question by examining compelling performance measures for two broad applications of practical interest, namely radar dynamic spectrum access \cite{martone2021closing,griffiths2014radar,Kirk2019} and multiple target tracking \cite{sira2009waveform,sira2006waveform,Shi2022}. We apply tools from decision theory and information theory to compare the performance of an online learning-based waveform selection algorithm with fixed rule-based strategies. More specifically, we show that the diagonal dominance and entropy rate of the underlying Markov transition model that describes the radar scene can be used to predict when an online learning strategy will stochastically dominate a fixed rule-based approach.

A well-noted concern associated with real-world applications of online learning algorithms is reliability \cite{heger1994consideration,keramati2020being}. While many learning algorithms are designed to perform well asymptotically and on average, oftentimes little is known about the worst-case performance or the entire distribution of possible outcomes. Thus, it is of particular importance to identify scenarios in which the use of learning algorithms is inherently high-risk. In this study, we examine the entire distribution of outcomes associated with particular decision strategies.

Our findings demonstrate that the entropy rate and diagonal dominance of Markov transition matrices can be used as strong general predictors for the performance of learning-based and rule-based waveform selection strategies, respectively. Additionally, we develop a process for establishing an ordering over waveform selection strategies when aspects of the physical scene are known or can be estimated. These results provide insight into the types of target, interference, and clutter behavior which are most and least amenable to the use of learning algorithms.

\subsection{Contributions}
We seek insight towards the general question ``when is online learning-based cognitive radar beneficial for adaptive waveform selection?" by examining the general framework of a closed-loop waveform selection problem and applying tools from stochastic dominance. We examine the performance of a Thompson Sampling based waveform selection scheme relative to a fixed baseline algorithm for two practical scenarios. Thompson sampling is chosen as a representative learning algorithm due to its widespread practical use, and theoretical performance which is close to the lower-bound for multi-armed bandit problems.

First, we examine a dynamic spectrum access scenario in which the radar selects a contiguous group of sub-channels to transmit in during each decision interval. We prove that given enough time, the TS strategy will dominate a fixed rule-based strategy in some sense. 

This paper expands on our preliminary work \cite{thornton2022cognitive} by establishing more general relationships between environmental characteristics and algorithm performance, which are applicable to a broad range of waveform selection and sensor scheduling problems under time-varying conditions. In particular, we consider a new scenario of multiple-target tracking, whereas our previous work \cite{thornton2022cognitive} only considered the dynamic spectrum access scenario. Additionally, this work considers \emph{stochastic dominance} relationships between waveform selection strategies, allowing for a comparison of the entire distribution of outcomes associated with a particular waveform selection policy. Our previous work considers only average performance metrics, which may be ineffective for applications like target tracking, where worst-case performance is often a crucial consideration.

The remainder of the paper is structured as follows. In Section \ref{se:problem}, we present the general state-space model for the cognitive radar problems considered herein. In Section \ref{se:dsa}, we establish tools to analyze the value of online learning algorithms in a dynamic spectrum access scenario. Section \ref{se:tracking} describes a more general scenario, in which the radar tracks multiple extended targets at a time. Section \ref{se:conclusions} provides concluding remarks and insights for practitioners.

\section{Problem Statement and Technical Motivation}
\label{se:problem}

\begin{figure*}
	\centering
	\includegraphics[scale=0.85]{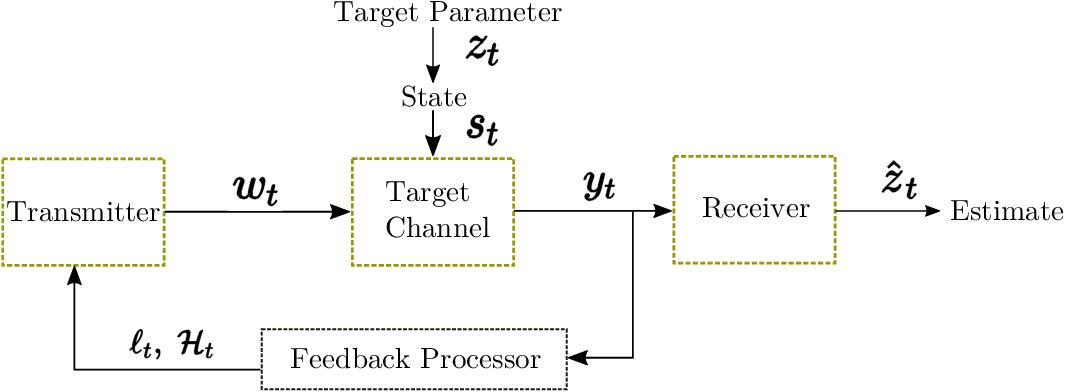}
	\caption{\textsc{Block Diagram} of general closed-loop waveform selection process. Each of the problems considered in this paper can be viewed as a special case of this set-up. Problem specific components include modeling the target channel, designing the feedback processor, and selecting a decision rule for the waveform selection itself.}
	\label{fig:block}
\end{figure*}

In this paper, we consider the following discrete-time waveform selection process, which is illustrated in Figure \ref{fig:block}. Each time step $t= \{1,2,...,n\}$, the radar selects a waveform $w_{t}$ from a finite library of candidates $\mathcal{W}$. The waveform is transmitted over a two-way propagation channel, possibly containing multiple scatterers, which is described by its \emph{state} $s_{t}$. A return signal $y_{t} \in \ncalY$ is then observed at the receiver. The relationship between the transmitted waveform, received signal, and channel state is captured by the two-way target channel probability mass function, 
\begin{equation}
	p(y_{t}|w_{t},s_{t}),
\end{equation} 
which is unknown to the radar \emph{a priori}.

Using information from the received signal $y_{k}$, the receiver produces an estimate $\hat{z}_{t}$ of some unknown target parameter vector $z_{t}$. This vector may contain multiple parameters from multiple targets. The target parameter vector is related to the channel state through the pmf $p(s_{t}|z_{t})$. The quality of the target parameter estimate is captured in an observation $o_{t}$.

Importantly, the true channel state is never directly observed and must be estimated from the history of waveforms and received signals before the current time slot, denoted by the set $\mathcal{H}_{t-1}$. The transmitter's goal is to select waveforms which maximize information about the target process $\{z_{t}\}$ using only the information from $\mathcal{H}_{t-1}$. 

This problem has natural implications for radar, especially in target tracking, which usually involves transmit waveform optimization when the channel is time-varying and unknown. A reasonable strategy is to pose the waveform selection process as a contextual bandit problem, in which various waveforms are selected and losses are observed. However, there are some drawbacks. 

\begin{enumerate}
	\item[(1.)] The period of exploration might result in highly suboptimal performance for a prolonged interval. 
	\item[(2.)] The appropriate choice of loss function and context features may be unclear. 
	\item[(3.)] Enough may be known about the application \emph{a priori} such that a simple rule-based decision scheme can be implemented.
\end{enumerate}

In this paper, we attempt to characterize scenarios of practical interest, in which learning approaches may or may not be appropriate. We discuss the relative merits of learning-based approaches for each application area. We investigate several channel models of interest for radar and communications while staying within the bounds of the general waveform selection problem described above.

\section{Dynamic Spectrum Access}
\label{se:dsa}
\subsection{System Model}
In this section, we consider a finite-horizon, discrete time, dynamic spectrum sharing scenario, where a tracking radar must share a fixed channel of bandwidth $B$ with one or more communications systems. Time is indexed by $t = 1,2,...,n$, where each time step corresponds to a radar pulse repetition interval (PRI). It is assumed that the radar must transmit in every PRI. During each PRI, the radar has access to a noisy observation of the interference in a fixed channel of bandwidth $B$. The shared channel is divided into $N_{SB}$ equally-spaced sub-channels. The radar may transmit in any \emph{contiguous} grouping of sub-channels.

Let the interference state be represented by a $d$-element binary vector expressed by $s_{t} \in \ncalS \subseteq \{0,1\}^{d}$. The frequency content of the radar's waveform is similarly given by the binary vector $w_{t} \in \ncalW \{0,1\}^{d}$. Since the radar is limited to transmitting in contiguous sub-channels, the total number of allowable waveforms is $|\ncalW| = d(d+1)/2$. The radar's observation is an estimate of the current interference state $o_{t} \in \{0,1\}^{d}$. The interference channel behavior is completely specified by initial state $s_{0}$ and transition probability matrix $P$, with elements $p_{ij} = \mathbb{P}(s_{t} = j |s_{t-1} = i) \in [0,1]$. We note that the transition probability matrix, and therefore interference channel behavior, are unknown to the radar \emph{a priori}. Based on the underlying interference state and the choice of radar waveform, the radar receives a real-valued loss at each PRI, denoted by $\ell_{t}$. The exact specification of the loss mapping will be developed below.

\begin{assum}[Markov Interference Channel]
	We assume the interference state generating process can be represented by a stochastic matrix $P: \ncalS \times \ncalS \mapsto [0,1]$ composed of transition probabilities associated with specific state transitions. 
	\label{assum:Markov}
\end{assum}

Under Assumption \ref{assum:Markov}, the radar's choice of waveform does not impact the state transition behavior. However, the state transition model can be generalized to take form $P(s,a,s'): \ncalS \times \ncalA \times \ncalS \mapsto [0,1]$, which is the usual model for a Markov Decision Process (MDP). The MDP model is applied to cognitive radar problems in \cite{charlish2020implementing,selvi2020reinforcement}. The model can also be further generalized to incorporate any $K^{\text{th}}$ order Markov process, as in \cite{thornton2022universal}. However, as the model grows in complexity, it becomes increasingly difficult to estimate $P$ accurately.

\begin{Def}[Diagonal Dominance]
	A transition matrix $P$ is said to be diagonally dominant if the diagonal elements of the transition matrix are larger than the corresponding row sum of off-diagonal elements the matrix. Mathematically, this can be expressed as
	\begin{equation}
		|P(m,m)| \geq \sum_{m \neq n} |P(m,n)| \quad \forall m.
	\end{equation}
\end{Def}

Now that the waveform selection process has been introduced, we describe important events which can be used to evaluate the performance of a waveform selection policy. These events are \emph{collisions} and \emph{missed opportunities}, which are defined as follows: 

\begin{Def}[Collisions]
	Let $C_{t}$ correspond to a `collision' event, in which $w_{t}$ and $s_{t}$ contain energy in overlapping frequency bands during PRI $t$. More precisely, 
	\begin{equation}
		C_{t} = 
		\begin{cases}
			1, & \text{if } \sum_{i=1}^{d}\mathbbm{1}\{w_{t}[i] = s_{t}[i] \} > 0 \\
			0, & \text{if } \sum_{i=1}^{d}\mathbbm{1}\{w_{t}[i] = s_{t}[i] \} = 0.
		\end{cases}
	\end{equation}
	Further, we define $N_{c,t} = \sum_{i=1}^{d}\mathbbm{1}\{w_{t}[i] = s_{t}[i] = 1 \}$ to be the number of colliding sub-bands at PRI $t$.
\end{Def}

\begin{Def}[Missed Opportunities]
	Let $M_{t}$ correspond to a `missed opportunity' event, which occurs whenever the radar does not make use of available sub-channels. More precisely, this event is defined by
	\begin{equation}
		M_{t} = 
		\begin{cases}
			1, & \text{if } |w^{*}| - |w_{t}| > 0 \\
			0, & \text{otherwise} 
		\end{cases}
	\end{equation}
	where $w^{*}$ is the widest available grouping of sub-channels in the current interference state vector $s_{t}$, and $w_{t}$ is a binary vector which identifies the frequency bands utilized by the waveform actually transmitted by the radar at PRI $t$. Further, we define $N^{mo}_{t} = \max{\{|w^{*}| - |w_{t}|,0\}}$ to be the number of available sub-channels missed at PRI $t$.
\end{Def}

We now define the loss function for the learning algorithm $\ell_{t}: \ncalS \times \ncalW \mapsto [0,1]$, which is specified by
\begin{equation}
	\ell_{t} = \left\{
	\begin{array}{lr}
		1, & \text{if } N^{c}_{t} > 0\\
		\eta N^{mo}_{t}, & \text{if } N^{c}_{t} = 0
	\end{array}
	\right\}, 
	\label{eq:loss}
\end{equation}
where $\eta \in [0,1/N]$ is a parameter which controls the radar's preference for bandwidth utilization compared to SIR. This loss function is justified as follows. Collisions are likely to cause a missed detection, numerous false alarms, or performance degradation to neighboring systems. Thus, collisions are to be avoided at all costs. Given that collisions can be effectively avoided, the radar wishes to utilize as much available bandwidth as possible to improve resolution performance and ensure efficient utilization of the available spectrum. 

We note that the loss function expressed in (\ref{eq:loss}) is not necessarily unique. However, this formulation balances the trade-off between the two fundamental errors for this application in a manner that expresses the importance of each. In the numerical results, we note that the value of $\eta$ primarily impacts the radar's behavior at inflection points. For example, when it is unclear whether the radar should attempt to utilize a particular band or not, the value of $\eta$ dictates the radar's level of `risk aversion', effectively speaking. In Section \ref{ss:numerDsa}, we examine the impact of $\eta$ on radar behavior numerically.

\subsection{Waveform Selection Schemes}
The baseline rule-driven adaptive approach considered in this section is the sense-and-avoid (SAA) policy, a form of which is implemented using low-cost hardware in \cite{Kirk2019,kirk2020performance} and is defined here as follows:
\begin{Def}[Sense-and-Avoid Policy]
	Let $\pi_{\text{SAA}}: \ncalO \mapsto \ncalW$ be a fixed decision function which selects the widest contiguous bandwidth available assuming $s_{t} = o_{t-1}$. The decision function is implemented algorithmically as follows. 
	\begin{enumerate}
		\item The position of the longest run of zeros in vector $o_{t-1}$ is found.
		\item Waveform $w_{k}$ is selected such that the radar transmits in the widest available bandwidth.
		\item If there are multiple vacancies of the same length, ties are broken randomly.
	\end{enumerate}	 
	\label{def:saa}
\end{Def}

\begin{Def}[Measure of Diagonal Dominance]
	To quantify the diagonal dominance of a transition matrix $P$, we may compute the sample correlation coefficient $R \in [-1,1]$ between the rows and columns of $P$. Details of this computation are provided in Appendix \ref{se:dominance}.
\end{Def}

\begin{theorem}[Bound on the Expected Loss of SAA]
	We may bound the $n$-step expected loss of the SAA policy as follows
	\begin{align}
		\E[\ell^{n}(\pi_{\text{SAA}})] &  \leq n \sum_{s_{i} \in \mathcal{S}} V_{n}(s_{i}) \sum_{s_{j} \neq s_{i}} P(i,j) \\
									   &  \leq n \sum_{s_{i} \in \ncalS}  V_{n}(s_{i}) [ 1-P(i,i) ]
	\end{align}
	where $V_{n}$ is the expected number of visits to state $s_{i}$ over $n$ time steps. As $n \rightarrow \infty$, given the existence of a stationary distribution on $P$, we have
	\begin{align}
		\E[\ell^{n}(\pi_{\text{SAA}})] &\leq n \sum_{s_{i} \in \ncalS} \mu_{i} (1-p_{ii}) \\
								       & \approx n (1-R) \label{eq:rApprox}
	\end{align}
	where $\{\mu_{i}\}_{s_{i} \in \ncalS}$ is the stationary distribution of the underlying state-generating Markov chain. We note that (\ref{eq:rApprox}) is a valid approximation when the stationary distribution is nearly uniform.
	\label{thm:lossSaa}
\end{theorem}

\begin{figure*}
	\centering
	\includegraphics[scale=0.55]{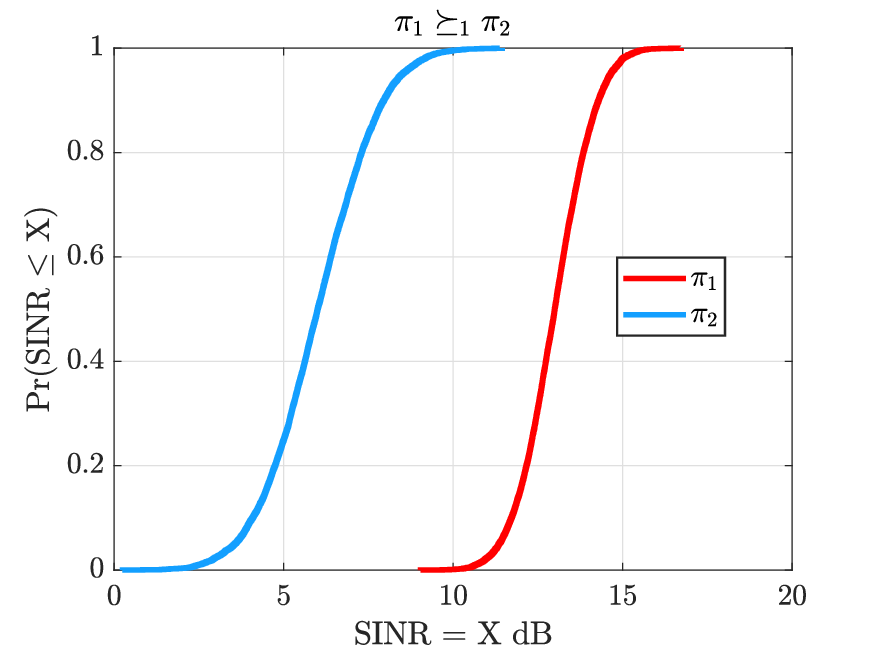}
	\includegraphics[scale=0.55]{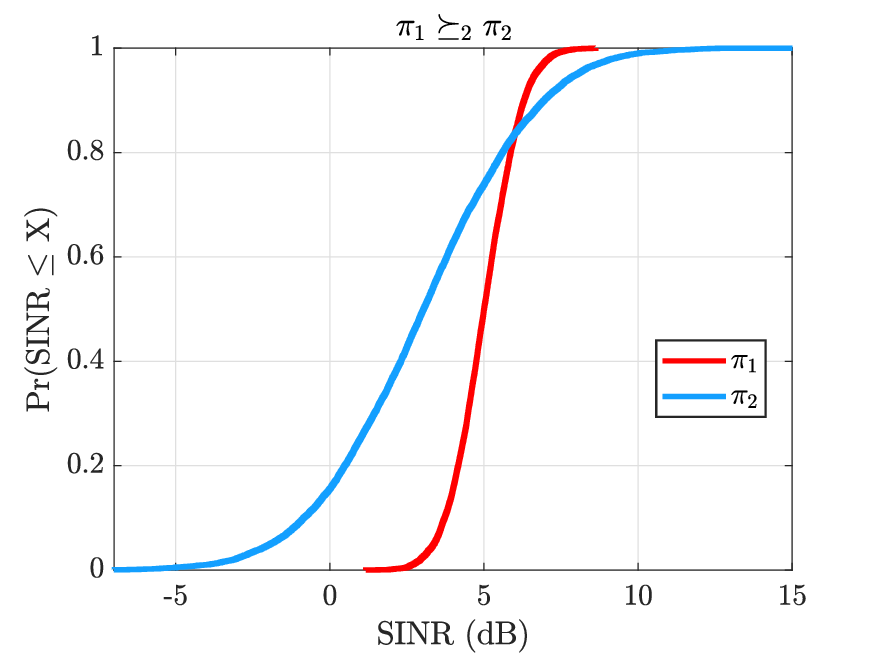}
	\caption{Example of first and second-order stochastic dominance for two waveform selection strategies. The performance metric of interest is observed SINR, and the loss is assumed to be inversely proportional to the SINR.}
	\label{fig:example}
\end{figure*}

Thus, we observe that the performance of SAA is dictated by the values off-diagonal elements of $P$ in states which occur with high frequency (ie. $\mu_{i} >> 0$). We may expect good average performance from $\pi_{\text{SAA}}$ when the proportion of self-transitions is high, or $\sum_{i} \mu_{i} p_{ii} \rightarrow 1$. However, we note the bound in Theorem \ref{thm:lossSaa} is loose, as we have only used the fact that $\ell \leq 1$ to create the bound. In specific cases, the structure of the loss function and transition matrix may be exploited and the approximation can be made tight. The example below demonstrates this.

\begin{example}[Two-state Markov Chain]
	Consider a two-state MC with transition matrix $P$. Let $s_{1} = [1,1,0,0,0]$ and $s_{2} = [1,1,1,0,0]$. Then, under $\psaa$, the state transition $s_{1} \rightarrow s_{2}$ corresponds to a collision and $s_{2} \rightarrow s_{1}$ corresponds to a missed opportunity. We may then compute the n-step loss to be
	\begin{equation}
		\ell^{n}(\psaa) \approx n (\mu_{1} p_{12} + \mu_{2} p_{21} \eta) \quad \text{as} \; \; n \rightarrow \infty
	\end{equation}
\end{example}

A wide variety of learning-based algorithms can be employed for online waveform selection \cite{lattimore2020bandit}. We note that algorithms designed to solve a Markov decision process, such as $Q$-learning or deep $Q$-learning may provide near-optimal performance over a long time horizon. However, the noted issue of sample-efficiency can be a serious practical hurdle \cite{selvi2020reinforcement,thornton2020deep}.

For comparison's sake, we focus predominantly on the Thompson Sampling (TS) approach, which is useful due to the simplicity of implementation, impressive empirical performance, and theoretical performance guarantees. The TS policy involves picking the waveform which has the highest posterior probability of yielding the lowest loss, and was introduced in \cite{thornton2021constrained}. 

\begin{Def}[Thompson Sampling Policy]
	The Thompson Sampling (TS) policy selects the waveform $w^{*}$ which has the highest posterior probability of yielding the lowest loss. More precisely, let $\ncalH$ be the set of waveforms, losses, and observations seen by the radar until PRI $t$, and let $\theta \in \Theta$ be a parameter employed by the radar to account for unknown time-varying effects of the channel. The TS policy then proceeds as follows,
	\begin{multline}
		\pi_{\text{TS}} = \\ \int_{\Theta} \mathbbm{1} \left[\mathbb{E}\left[\ell \mid w^*, o, \theta\right]=\min _{w^{\prime}} \mathbb{E}\left[\ell \mid w^{\prime}, o, \theta\right]\right] P(\theta \mid \mathcal{H}) d \theta
		\label{eq:ts}
	\end{multline}
\end{Def}

Details of a practical algorithm to implement $\pi_{\text{TS}}$ can be found in \cite{agrawal2013thompson}.

\subsection{Performance Characterization}

\begin{Def}[Statewise Dominance]
	A waveform selection policy $\pi_{1}$ is said to dominate policy $\pi_{2}$ \emph{statewise} if $\ell(\pi_{1}(s)) \leq \ell(\pi_{2}(s))$ for every $s \in \ncalS$.
\end{Def}

\begin{remark}
	Utility maximizing strategies, such as the Bellman optimal policy specified in (\ref{eq:Bellman}), often sacrifice performance in states that occur with low probability, in order to improve performance in states that occur with high probability. Thus, we argue that while statewise dominance is a strong condition, it is not always the most meaningful notion of dominance for cognitive radar problems.  
\end{remark}

\begin{Def}[Long-Term Average Dominance]
	Let $\lambda_{\pi_{1}}$ be the long-term average cost associated with acting according to policy $\pi_{1}$ in some fixed environment $E$, specified by Markov transition kernel $P$. Policy $\pi_{1}$ is said to dominate $\pi_{2}$ on average if $\lambda_{\pi_{1}} < \lambda_{\pi_{2}}$.
\end{Def} 

\begin{Def}[First-Order Stochastic Dominance (FOSD)]
	Policy $\pi_{1}$ is said to be first-order stochastic dominant over policy $\pi_{2}$, or $\pi_{1} \succeq_{1} \pi_{2}$, if 
	\begin{equation}
		P(\ell(\pi_{1}) \geq x) \leq P(\ell(\pi_{2}) \geq x)
	\end{equation}
   for every $x \in \mathbb{R}$ and $P(\ell(\pi_{1}) \geq x) < P(\ell(\pi_{2}) \geq x)$ for some $x$. Since we are using losses as a performance measure, we apply the convention that smaller outcomes are preferable. An example of first-order stochastic dominance is seen in Figure \ref{fig:example}.
\end{Def}

\begin{Def}[Second-Order Stochastic Dominance (SOSD)]
	Let $F_{1}$ be the CDF of a performance statistic, namely here the observed loss as specified in (\ref{eq:loss}), under policy $\pi_{1}$ and likewise $F_{2}$ be the cdf under $\pi_{2}$. Then $\pi_{1}$ is said to second-order dominate $\pi_{2}$ if 
	\begin{equation}
		\int_{-\infty}^x\left[F_1(t)-F_2(t)\right] d t \geq 0 
	\end{equation}
	for all $x \in \nbbR$ with strict inequality at some $x$.
\end{Def}

We note that second-order stochastic dominance is a necessary condition for first-order stochastic dominance. Thus, first-order stochastic dominance is a stronger condition. It has been shown that any risk-averse rational decision maker prefers policy $\pi_{1}$ to $\pi_{2}$ when $\pi_{1} \succeq_{2} \pi_{2}$. Thus, 

\begin{theorem}
	For any time-varying Markov interference channel, namely whenever some non-diagonal elements of $P$ are less than one, $\pi_{\text{TS}}$ weakly dominates dominates $\psaa$ in the long-term average sense. Further, as $t \rightarrow \infty$, $\pts$ dominates $\psaa$ in the second-order stochastic sense.
\end{theorem}

\begin{proof}
	See Appendix
\end{proof}

\begin{remark}
	For time-varying Markov interference channels that are quickly varying, namely, when the diagonal elements of $P$ are close to zero,
\end{remark}

\begin{Def}[Entropy Rate]
Given state transition probability matrix $P$, the entropy rate of the state generating process is
\begin{equation}
	H(S) = -\sum_{i,j} \mu_{i} P(i,j) \log P(i,j),
\end{equation}
where $\mu$ is the stationary distribution of the underlying Markov chain.
\end{Def}

From the definition of missed opportunities and collisions, we note that for a fixed decision rule such as SAA, both events can be associated with specific state transition events. Given that the SAA policy is a fixed decision rule, the frequency of these events can be identified by the following decomposition of $P$:

\begin{prop}[Structure of Transition Matrix]
	Define three sets $\ncalA$, $\ncalB$, and $\ncalC$. Set $\ncalA$ consists of transitions $t_{ij}$ such that $|\pi_{\text{SAA}}(s_{i})| < |\psaa(s_{j})|$ and such that $|\psaa(s_{i})| = |\psaa(s_{j})|$ but $\psaa(s_{i}) \neq \psaa(s_j)$. In other words, the set of transitions such that acting on the SAA policy will produce a collision. Set $\ncalB$ consists of the transitions $t_{ij}$ such that $|\psaa(s_{i})| > |\psaa(s_{j})|$. Finally, set $\ncalC$ consists of $t_{ij}$ such that $\pi_{\text{SAA}}(s_{i}) = \psaa(s_{j})$.
\end{prop}

\begin{remark}[Effectiveness of SAA]
	If $s_{t} = s_{t-1}$ and $o_{t} = o_{t-1} = s_{t}$, then transmitting $w_{t} = \pi_{\text{SAA}}(o_{t-1})$ during each PRI results in zero collisions and zero missed opportunities. Equivalently, the SAA strategy is effective only when the diagonal elements of $P$, $p_{ii}$, are close to one, and when the probability of being in set $\ncalC$ is large.
\end{remark}

\begin{figure*}
	\centering
	\includegraphics[scale=0.55]{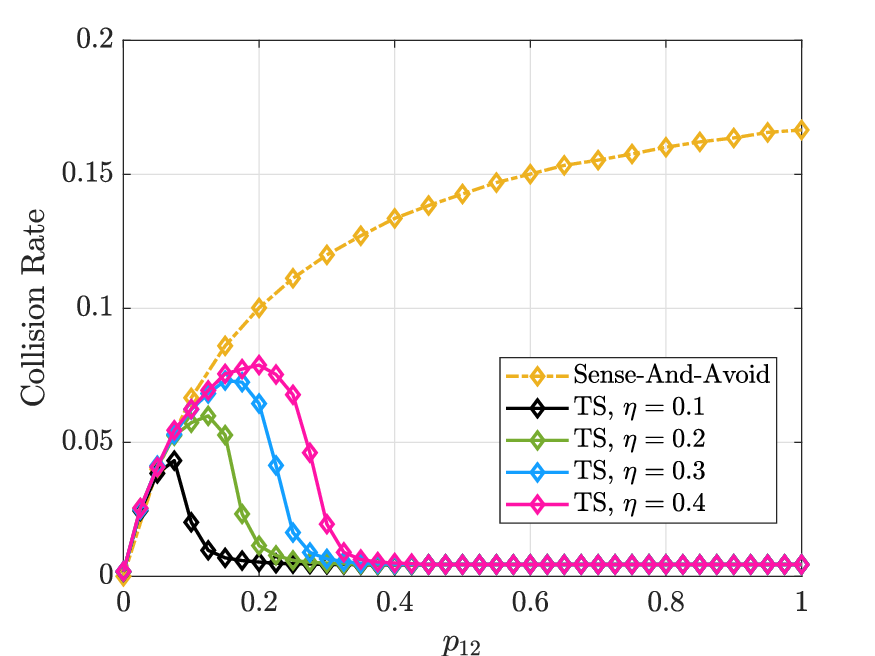}
	\includegraphics[scale=0.55]{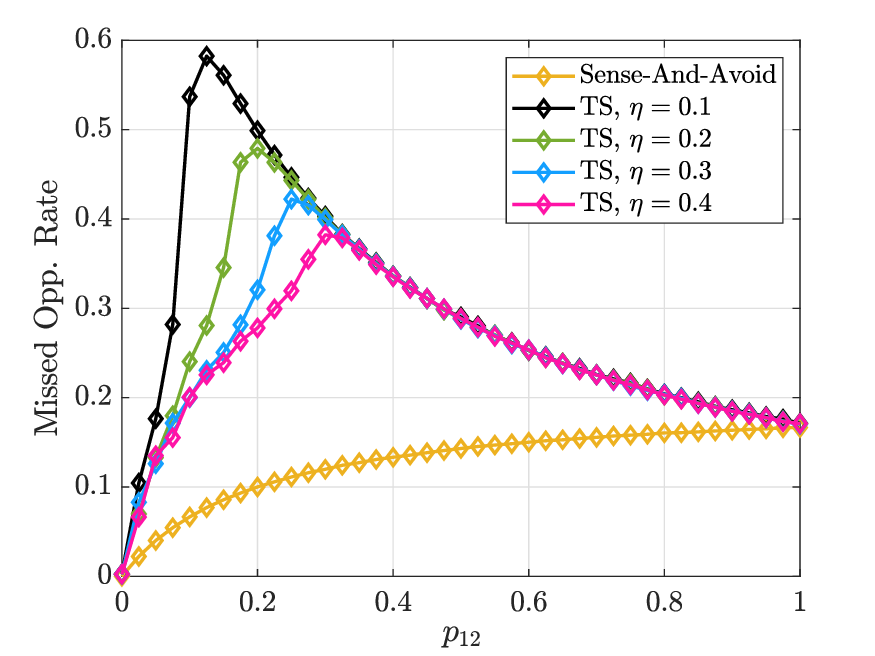}
	\caption{Analysis of the collision and missed opportunity rates as transition probability $p_{12}$ is varied. We observe that as the interference becomes increasingly dynamic ($p_{12} \rightarrow 1$), the learning-based approaches definitively outperform SAA over a time horizon of $n = 1e4$.}
	\label{fig:firstCase}
\end{figure*}

\begin{figure*}
	\centering
	\includegraphics[scale=0.55]{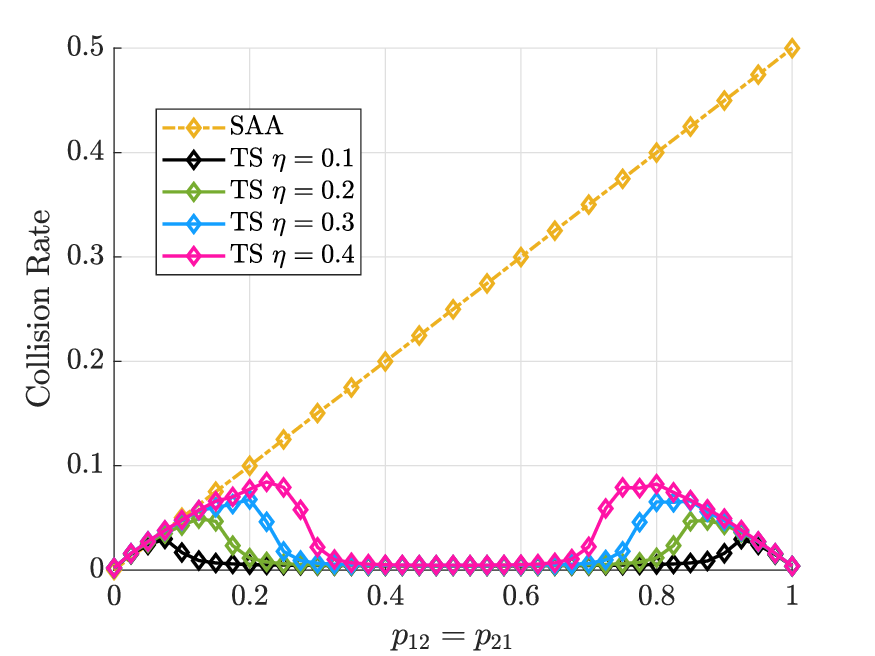}
	\includegraphics*[scale=0.55]{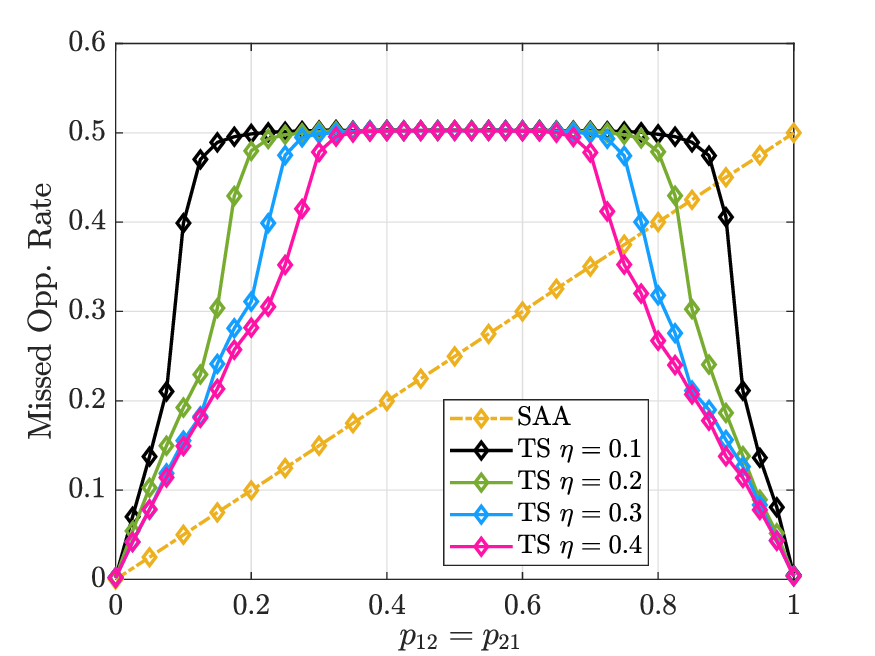}
	\caption{Collision and missed opportunity rates as transition probabilities $p_{12}$ and $p_{21}$ are jointly varied. As the interference becomes increasingly dynamic and deterministic ($p_{12} = p_{21} \rightarrow 1$), the learning-based approaches definitively exceed SAA in performance.}
	\label{fig:joint}
\end{figure*}

\begin{remark}[Effectiveness of Learned Policies]
	Learned policies will be effective under the following conditions: 1. Long time horizon $n \rightarrow \infty$, 2. Many elements of $P$ close to $0$ or $1$ (nearly deterministic), and 3. When the probability of being in set $\ncalA$ or $\ncalB$ is large.
\end{remark}

Having generally identified under which conditions the SAA and TS policies will be effective, we now move towards a more specific characterization, employing tools from dynamic programming and statistical decision theory.
\begin{Def}[Optimal Long-Term Average Cost]
	Given an initial state $s_{0}$, the optimal long-term average cost achievable by a policy is given by
	\begin{align}
		&\lambda^*\left(s_{0}\right) = \nonumber\\
		&\inf _{\pi \in \Pi} \limsup _{n \rightarrow \infty} \mathrm{E}_\nu\left[\frac{1}{n} \sum_{t=1}^n \ell\left(s_t, w_t \right) \mid s_{1}^{t-1}, w_{1}^{t-1}\right],
	\end{align}
	where the infimum is taken over the set of all admissible decision rules $\Pi$.
\end{Def}

A policy which achieves the optimal long term average cost is known as the Bellman optimal policy, which is not necessarily unique.
\begin{Def}[Bellman Optimal Policy]
	\begin{align}
		\label{eq:Bellman}
		\pi^{*} = \min _{w_B} \sum_{s_{t+1}} P &\left(s_{t+1} \mid s_{1}^{t}\right) \\ \nonumber
		& \times\left[\ell\left(s_t+1, w_B \right)+\alpha \ell\left(s_1^{t}, w_1^{t}\right)\right],
	\end{align}
	where $\alpha \in [0,1]$ is a weighting factor called the \emph{discount rate}.
\end{Def}

\subsection{Numerical Results}
\label{ss:numerDsa}
In these simulations, we examine the performance of the Thompson Sampling-based waveform selection policy with an uninformative Gaussian prior, implemented using the algorithm described in \cite{thornton2021constrained}. Our baseline rule-based approach is the SAA algorithm described in Definition \ref{def:saa}. We compare these algorithms over a wide variety of simulated scenarios, and consider both average performance metrics (Figures \ref{fig:firstCase}-\ref{fig:lown}) and the distribution of outcomes (Figures \ref{fig:hrcol}-\ref{fig:sosd}).

Firstly, we note that the dynamic spectrum access scenario considered in this section presents a largely abstracted `high-level' simulation, while more of the specific radar signal processing details are considered in the multi-target tracking scenario. In the first set of experiments, we consider average performance metrics over $500$ runs of each algorithm, where each run consists of $n=1e4$ time-slots, or radar pulse repetition intervals.

For the dynamic spectrum access scenario, our simulation environment largely abstracts the finer signal processing details, and we focus on the fundamental quantities of collisions and missed opportunities. In Figures 3-7, we consider very simple Markov interference behavior which can be thought of as a standard `Gilbert-Elliot' model \cite{Hochwald1999}, where an interference source with fixed location in the frequency domain can be thought of as switching from `on' to `off' and vise versa according to the transition probabilities of the underlying Markov chain.

The main performance metrics evaluated in this section, namely collisions and missed opportunities, reflect the efficient usage of available spectrum, and are expected to correlate strongly with SINR and range resolution, respectively. We note that an emphasis on these metrics is not unique to the present work, and has become commonplace in the radar dynamic spectrum access literature. For example, spectral collisions and missed opportunities have been studied in \cite{Kirk2019,Kirk2023,selvi2020reinforcement,thornton2020deep,kovarskiy2020evaluation}.

In the first set of experiments, we begin by considering a simple Markov interference channel with only 2 states, having stationary transition probability matrix $P = [p_{11},p_{12};p_{21},p_{22}]$. The initial state $s_{0}$ is selected randomly. We note that under $\pi_{\text{SAA}}$, the transition associated with $p_{12}$ corresponds to a collision event, while $p_{21}$ corresponds to a missed opportunity. The probability of missed opportunity and collision under $\pi_{\text{SAA}}$ are then 
\begin{align}
	P(\text{Col}) &= \mu_{1} p_{12} = \frac{p_{12}p_{21}}{p_{12}+p_{21}},	
\end{align}
and
\begin{align}
	P(\text{Mo}) &= \mu_{2} p_{21} = \frac{p_{12}p_{21}}{p_{12}+p_{21}}.
\end{align}

Thus, if we fix $p_{11}$, $p_{21}$, and $p_{22}$ while varying $p_{12}$, we expect the probability of both missed opportunity and collision under $\pi_{\text{SAA}}$ to tend towards $p_{21}/(1+p_{21})$ sub-linearly as $p_{12} \rightarrow 1$. If we send $p_{12} = p_{21} \rightarrow 1$, then we expect the rates of missed opportunity and collision under $\pi_{\text{SAA}}$ to tend towards $1$ linearly.

In Figure \ref{fig:firstCase}, we observe the collision and missed opportunity rate over a time horizon of $n = 1e4$ as transition probability $p_{12}$ is varied and the other transition probabilities remain fixed, noting that the expected trends in SAA performance occur. We observe that for low values of $p_{12}$, namely less than $0.18$, SAA and TS perform nearly equivalently in terms of collision rate. However, in this regime, SAA provides some advantage as there are a reduced number of missed opportunities. We note that an inflection occurs around $p_{12} = 0.3$. As the channel becomes less stationary, the TS policy learns to avoid collisions effectively by avoiding certain actions altogether. Equivalently, the band is utilized more often, and the number of opportunities missed by the avoidant scheme of TS decreases. However, in the regime of $0.4 \leq p_{12} \leq 0.6$, the TS policy opts to avoid all collisions at the expense of considerable missed opportunities due to difficulty in predicting utilization in the band of interest. Thus, we observe that as the interference becomes more time-varying, and $P$ becomes less diagonally dominant, TS performs much better than SAA, as expected.

In Figure \ref{fig:joint}, we observe the performance of each scheme as transition probabilities $p_{12}$ and $p_{21}$ are jointly varied between $0$ and $1$. Values of $p_{12} = p_{21}$ near zero indicate a nearly stationary channel, while values $p_{12} = p_{21} \approx 1$ indicate a rapidly fluctuating interference channel that changes deterministically. We note the performance can be segmented into three regions. For $p_{12} = p_{21} \leq 0.2$, the performance of TS closely mirrors SAA in terms of collisions, while incuring a slightly higher missed opportunity rate. In the regime of $0.2 < p_{12} = p_{21} < 0.6$, the TS approach effectively mitigates collisions at the expense of a missed opportunity rate of approximately $0.5$. In the final regime, $p_{12} = p_{21} > 0.6$ we note that the TS strategy outperforms SAA in terms of both missed opportunities and collisions. This is because the interference transitions become closer to deterministic as the transition probabilities approach $1$.

\begin{figure}
	\centering
	\includegraphics[scale=0.5]{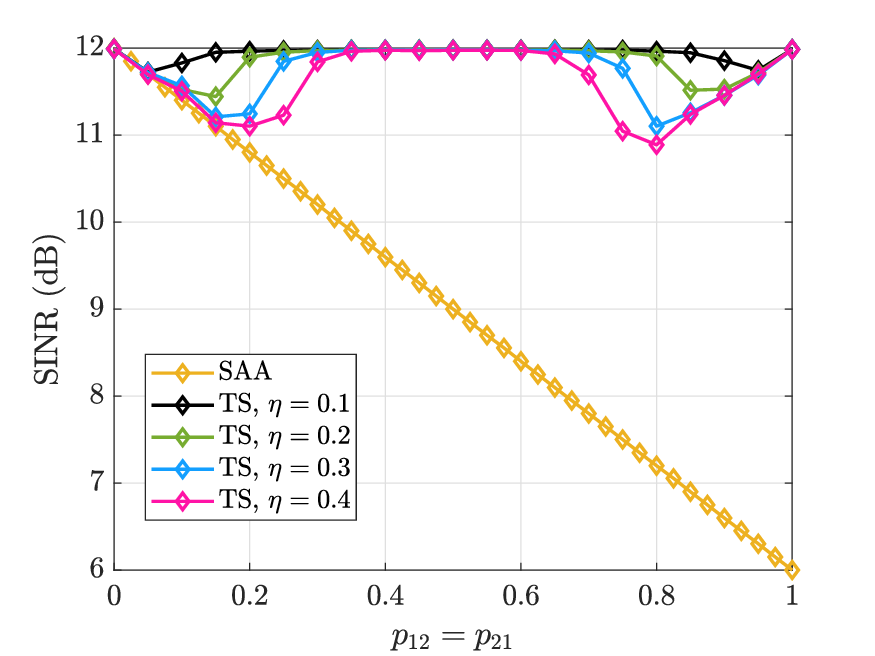}
	\caption{Average instantaneous SINR as transition probabilities $p_{21}$ and $p_{21}$ are jointly varied. As the channel becomes increasingly dynamic, the SINR under the SAA policy drops and reliability of target detection is diminished.}
	\label{fig:sinr}
\end{figure}

\begin{figure}
	\centering
	\includegraphics[scale=0.5]{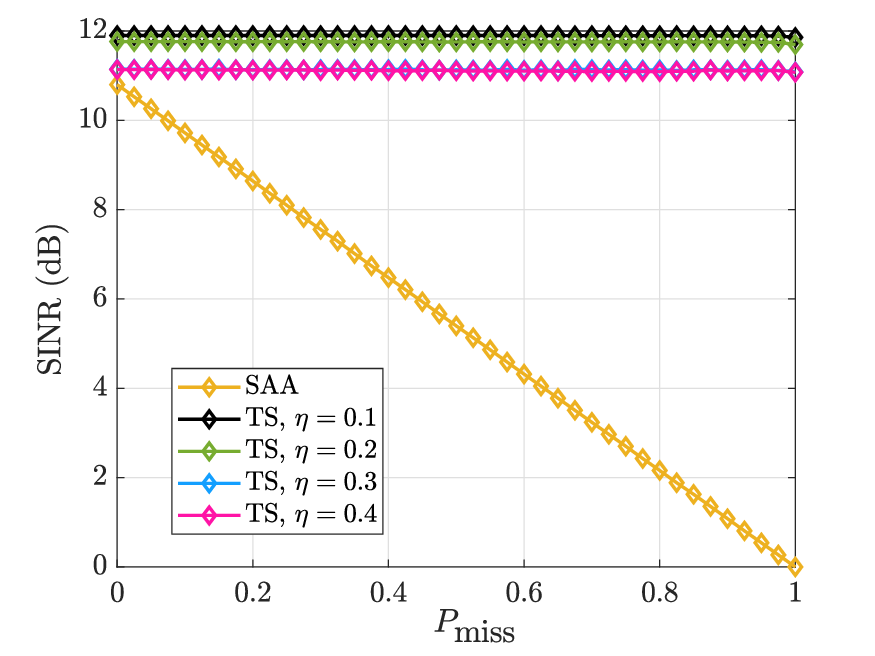}
	\caption{Impact of missed interference detections on average SINR. TS policies are generally robust to incorrect observations, while the SAA policy quickly degrades in performance with increasing $P_{miss}$.}
	\label{fig:miss}
\end{figure}

In Figure \ref{fig:sinr}, we see the received SINR\footnote{In the experiments, SINR is approximated using the well-known radar range equation, assuming an $\text{INR}$ of 14dB.} of these approaches for the case of jointly varied transition probabilities. We observe that for all values of $p_{12} = p_{21}$, the TS waveform selection strategy effectively maintains a high SINR, while SAA begins to perform poorly as the channel becomes increasingly time-varying. For rapidly varying channels, it can be concluded that SAA is untenable to maintain a sufficiently high probability of detection. Additionally, as the interference behavior becomes closer to deterministically varying, the TS approach is able to effectively predict transitions, and approaches missed opportunity and collision rates nearing zero.

Figure \ref{fig:miss} shows the performance degradation of SAA when the assumption of perfect interference observation is lifted. In this experiment, the radar fails to observe the interference state $s_{t}$ with probability $P_{\text{miss}}$ at each time slot, and instead observes a vector of all zeros. Once again, the time horizon is $n = 1e4$ and the transition behavior is held constant at $p_{12} = p_{21} = 0.3$. We observe that while TS is very robust to imperfect observations, the performance of SAA quickly degrades, as expected.

\begin{figure}
	\centering
	\includegraphics[scale=0.5]{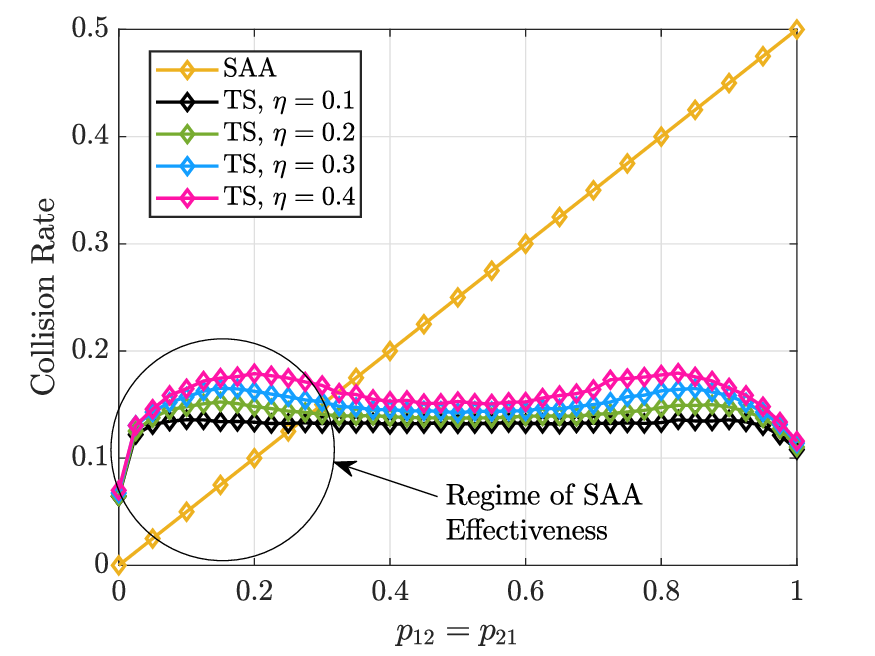}
	\caption{Collision rate of SAA and TS algorithms when time horizon is limited to $n=300$. We observe the regime where SAA is dominant.}
	\label{fig:lown}
\end{figure}

To highlight the regime where SAA is effective, we examine a drastically reduced time horizon, $n = 300$, in Figure \ref{fig:lown}. We observe that due to the learning time required by TS, the approach performs significantly worse than in the case of a longer time horizon, while the performance of SAA remains unchanged. This is to be expected as the regret bound of TS (ignoring logarithmic terms) is $\tilde{\mathcal{O}}(d^{3/2}\sqrt{n})$ \cite{agrawal2013thompson}. However, we note that even in the case of a short time horizon, when $p_{12} = p_{21} > 0.4$, TS performs better than SAA in terms of collision rate, due to the high degree of non-stationarity in the interference channel.

From these results, we observe that generally, the learned policies are more robust to quickly-varying interference behavior and imperfect interference observations than SAA. On the other hand, SAA performs decidedly better than SAA when the time horizon is short, interference observation probability is high, and the channel is varying relatively slowly. Thus, it is important for the system designer considering the use of a learning-based algorithm to take into account the expected time horizon of the learning process as well as state observation probability. In general, fixed decision rules will break down when fundamental assumptions are violated, while learning-based approaches may be more robust. We note that for many radar applications, such as target tracking and electronic warfare, the state observation probability may be low, suggesting the value of learning may be very high for applications in which specific domain knowledge is not available \emph{a priori}.

\begin{figure*}
	\centering
	\includegraphics[scale=0.5]{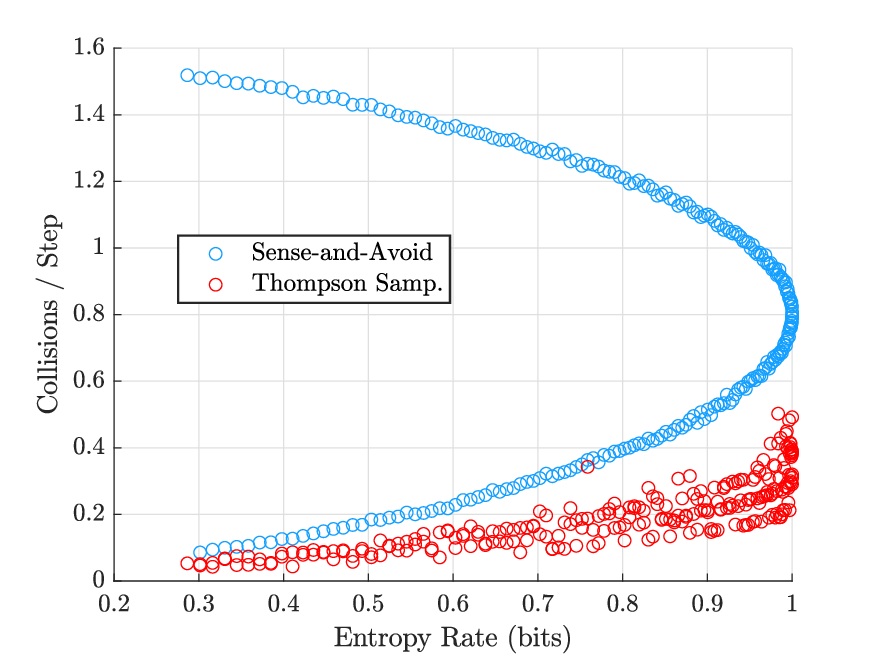}
	\includegraphics[scale=0.5]{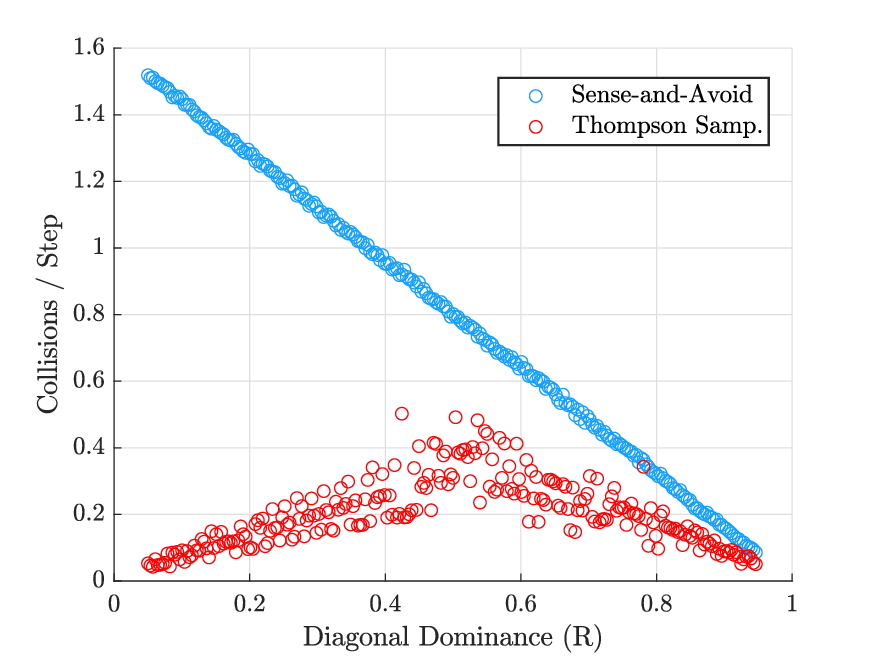}
	\caption{Collision Performance of randomly generated state transition matrices of size $(5,5)$. We observe a strong linear relationship between diagonal dominance and SAA performance, as well as a strong negative linear relationship between entropy rate and TS performance.}
	\label{fig:hrcol}
\end{figure*}

\begin{figure*}
	\centering
	\includegraphics[scale=0.5]{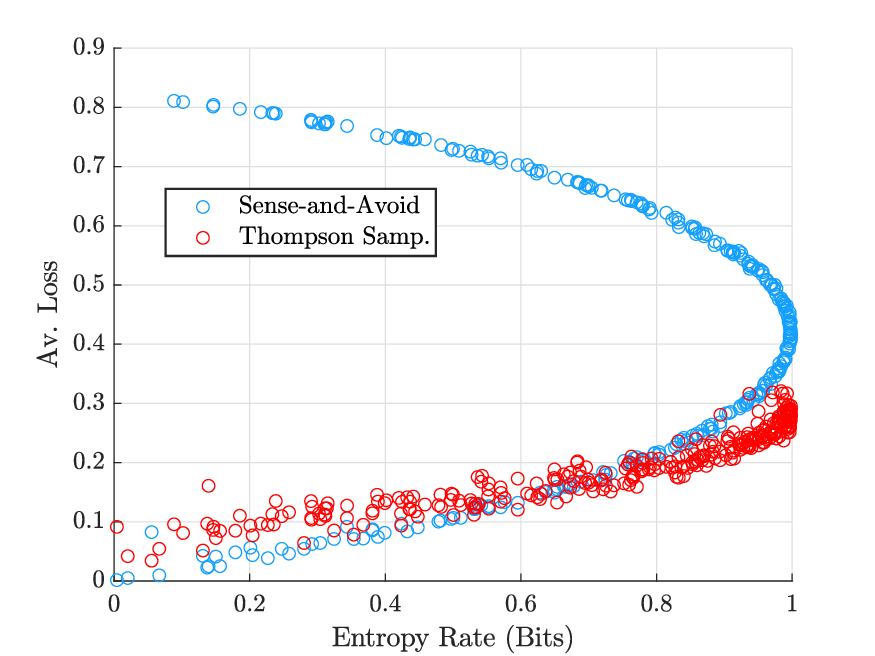}
	\includegraphics[scale=0.5]{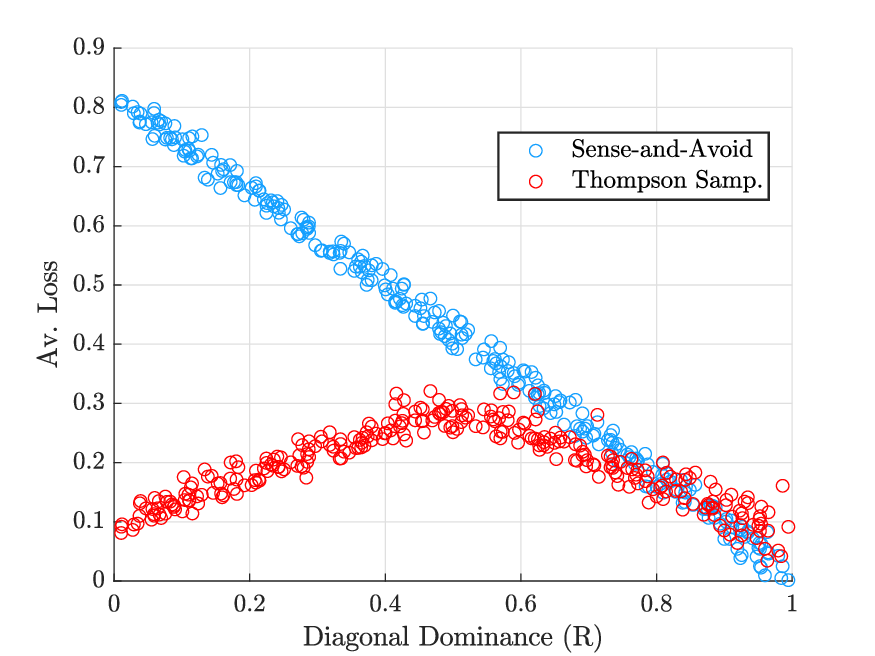}
	\caption{Variation in observed losses with entropy rate and diagonal dominance. We observe once again that for high entropy rates and low diagonal dominance values, TS reliably outperforms SAA even when the time-horizon is very limited.}
	\label{fig:hrloss}
\end{figure*}

In the next set of experiments, we examine the results of SAA and TS performance over a wider range of randomly generated transition matrices. From these results, we gain a greater understanding of the general performance trends in more practical cases of interest than the two-state Markov chain examined above.

In Figure \ref{fig:hrcol}, a set of $300$ randomly generated transition matrices of size $(5,5)$ are examined, and the average collision performance of the SAA and TS policies are compared with the entropy rate and diagonal dominance measure of each individual transition matrix. We observe that as diagonal dominance tends toward zero, the performance of SAA degrades linearly. On the other hand, TS is able to perform well for cases of low-diagonal dominance, provided that the entropy rate is low. We find that entropy rate is a much stronger predictor of TS performance than diagonal dominance, as expected given the recent information-theoretic analysis of TS.

Figure \ref{fig:hrloss} examines the same scenario as the previous experiment, but shows performance in terms of losses instead of collisions. Here we see that at very low entropy rates, as well as very high diagonal dominance values, TS performs marginally worse than SAA. This can be attributed to the reduced number of missed opportunities experienced by SAA in these regimes. We recall that the time horizon for both algorithms is limited to $n = 1e4$ steps, and expect that as the time horizon is increased, the performance of TS will tend towards the performance of SAA in these regimes. 

\begin{figure*}
	\centering
	\includegraphics[scale=0.5]{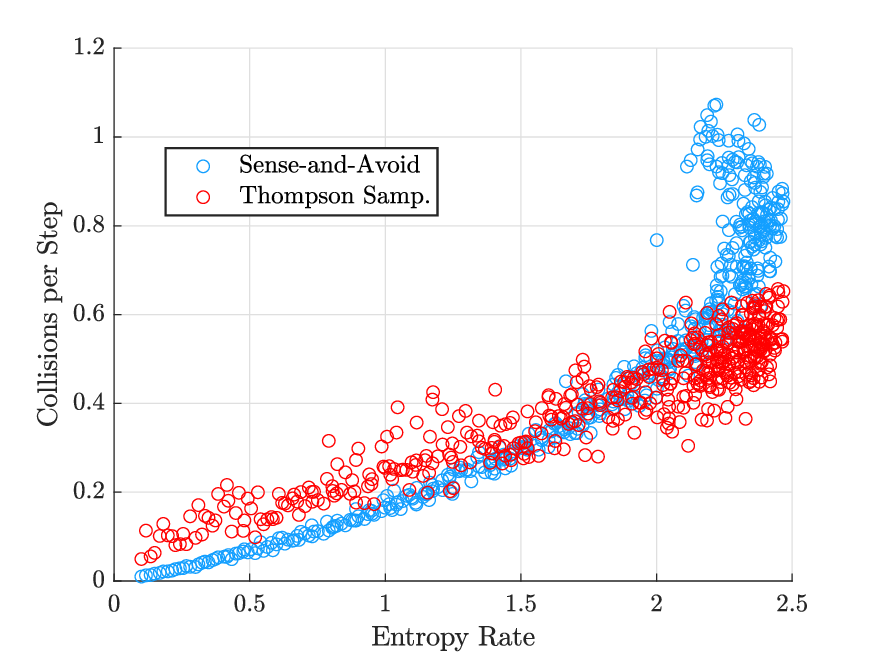}
	\includegraphics[scale=0.5]{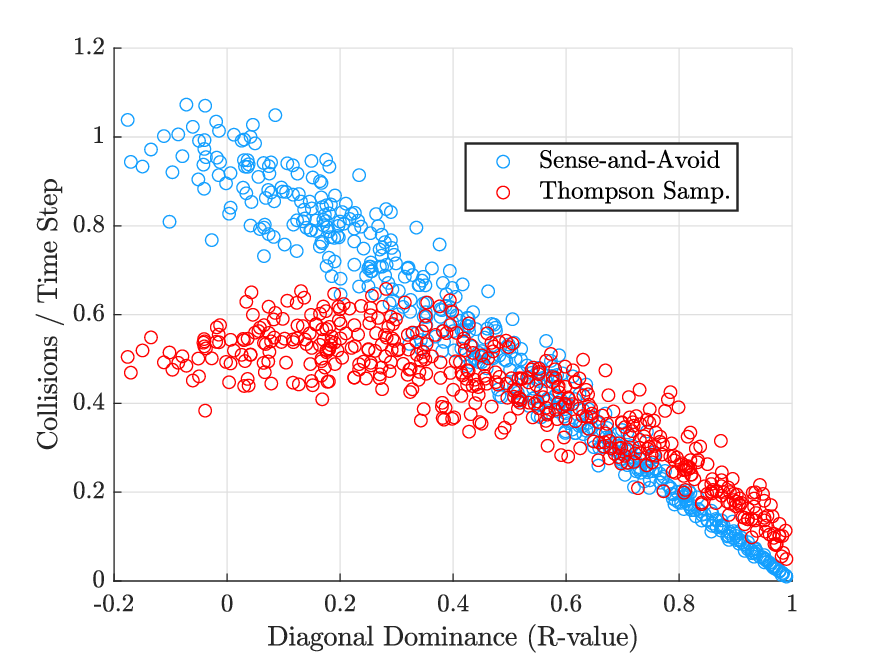}
	\caption{Collision performance for randomly generated larger state-transition matrices, having dimensionality $(125,125)$. We observe that entropy rate is a strong predictor of TS performance, and that diagonal dominance is a strong predictor of SAA performance. We observe that learning-based approaches are expected to be beneficial at high entropy rates and for low values of diagonal dominance.}
	\label{fig:hrdominant}
\end{figure*}

In Figure \ref{fig:hrdominant}, we observe the collision performance once again over a set of larger transition matrices, having dimensionality $125 \times 125$. In this scenario, we once again observe the same general performance trend. Namely, that the performance of TS degrades linearly as entropy rate increases, and that the performance of SAA degrades linearly as diagonal dominance increases. We observe that there is also a strong inverse relationship between entropy rate and SAA performance in this case. This can be attributed to the large size of the transition matrices. Since the transition probabilities are dispersed over many states, a high entropy rate generally implies low diagonal values. While both SAA and TS degrade as entropy rate increases, we observe that the degradation of TS is slower, and that there is a noticeable performance improvement at high entropy rates.

\begin{figure*}
	\centering
	\includegraphics[scale=0.5]{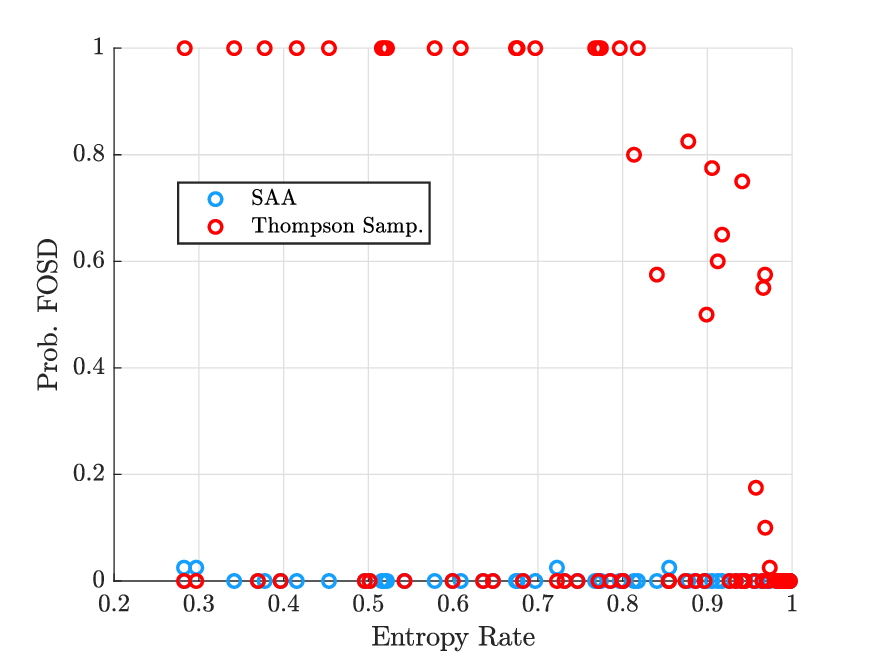}
	\includegraphics[scale=0.5]{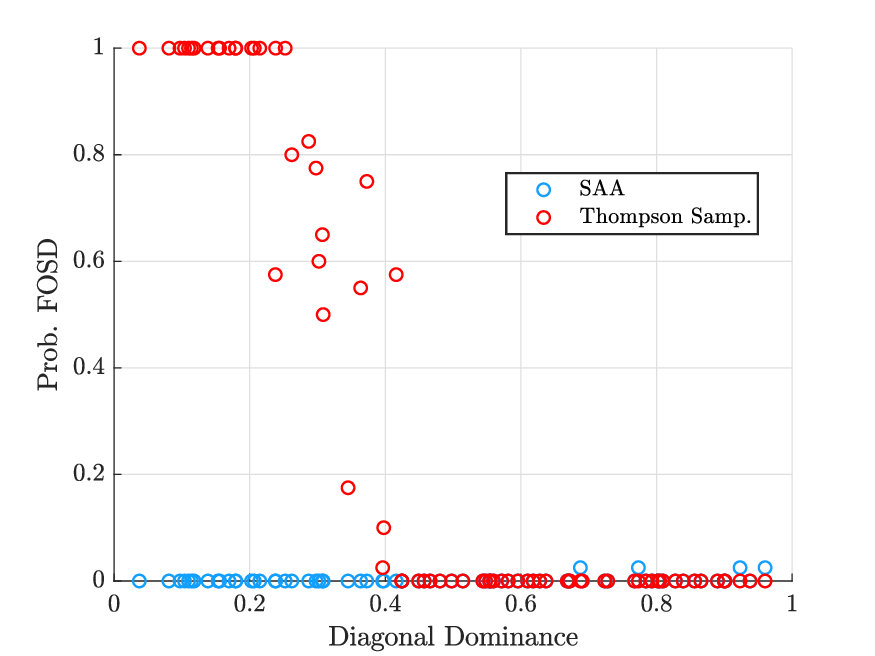}
	\caption{Comparison of the entropy rate and diagonal dominance to the probability that SAA and TS are FOSD with respect to the other algorithm. We observe that low diagonal dominance is the best predictor of when TS will be FOSD.}
	\label{fig:fosd}
\end{figure*}

\begin{figure*}
	\centering
	\includegraphics[scale=0.5]{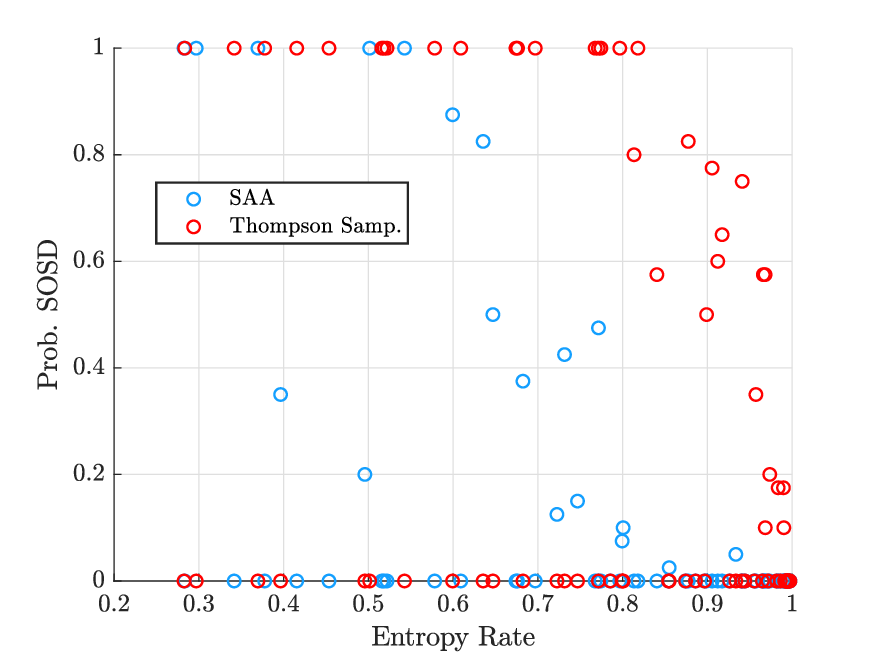}
	\includegraphics[scale=0.5]{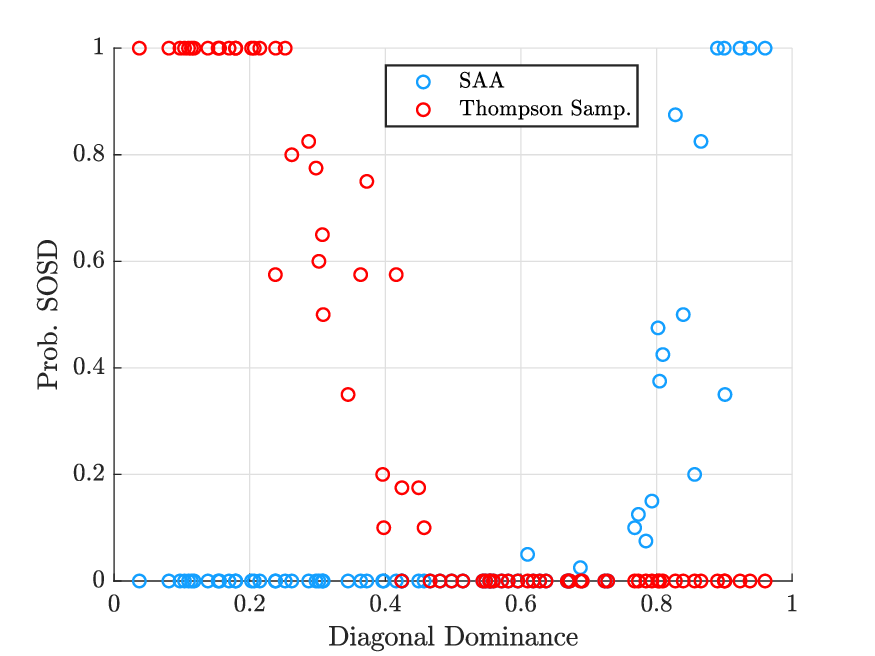}
	\caption{Second order stochastic dominance with varying entropy rate and diagonal dominance. We observe that diagonal dominance is an effective predictor of SOSD.}
	\label{fig:sosd}
\end{figure*}

In Figure \ref{fig:fosd}, we compare the entropy rate and digaonal dominance to the probability that SAA and TS will be FOSD with respect to the opposing algorithm. We note that SAA is only FOSD very infrequently at high values of diagonal dominance. On the other hand, TS is frequently FOSD at low diagonal dominance values. We observe that diagonal dominance is a much stronger predictor of FOSD, which is expected given that it is a stronger predictor of SAA performance than entropy rate. Likewise, Figure \ref{fig:sosd} examines the probability that each algorithm is SOSD with respect to the other. We observe that TS is predictably SOSD for diagonal dominance values around or below $0.4$, while SAA is frequently SOSD for diagonal dominance values around or above $0.75$. These results are expected, since at low diagonal dominance values SAA performs predictably poorly, and at high diagonal dominance values SAA performs predictably well, while the performance of TS at high diagonal dominance values depends on the time horizon and choice of prior distribution.

\section{Waveform-Agile Multiple Target Tracking}
In this Section, we study the online waveform selection problem in the context of generic multiple-target tracking, which is a more challenging scenario than the dynamic spectrum access problem studied in Section \ref{se:dsa}. 

In this application, the size of the state transition matrix $P$ becomes very large as the number of targets tracked increases. This can present a challenge for both learning-based and fixed rule-based waveform selection strategies, as the learning process can become cumbersome, and it may be difficult to select decision rules in advance.

The general framework of this problem is an extension of the model described in Section \ref{se:problem}. We once again consider an $n$-step discrete-time scenario where time is indexed by $t$. Here, the radar is tracking $M \geq 2$ targets. The state of a particular target $m$ at time step $t$ is described by the vector 
\begin{equation}
	s^{m}_{t} = [x^{m}_{t}, y^{m}_{t}, \dot{x}^{m}_{t}, \dot{y}^{m}_{t}], 
\end{equation}
where $x_{t}$ and $y_{t}$ are the two-dimensional position coordinates of target $m$, while $\dot{x}$ and $\dot{y}$ are the velocities in the $x$ and $y$ directions respectively. We note that if there are $P$ total position states, $V$ total velocity states, and $M$ targets, the size of the state space is $(P \times V)^{M}$. Thus, even for modest values of $M$, the size of the state transition matrix can be very large.

The true target states are unknown to the radar and must be inferred using noisy measurements. The observation of target $m$, denoted by $o^{m}_{t}$ is a noisy measurement of $s^{m}_{t}$. This observation is extracted from the resulting range-Doppler response. Since the scene contains $M$ targets, the underlying state of the entire scene is the composite vector of individual target states $s_{t} = [s^{1}_{t},...,s^{M}_{t}]$. Similarly, the composite observation of all the targets is denoted by $o_{t} = [o^{1}_{t},...,o^{M}_{t}]$. 

Each time step, the radar must select a waveform $w_{i}$, having complex envelope $f_{i}$, from a finite library $\mathcal{W}$ and observes a loss $\ell_{t}$. Here, the loss is an estimate of the SINR formed using the range-Doppler response. Let $\mathrm{RD}_{\text{dB}}$ be the range-Doppler map in dB units, indexed by range and Doppler cell, respectively. Then, the loss is computed as follows
\begin{equation}
	\ell_{t} = \frac{1}{|\mathcal{T}|}\sum_{i,j \in \mathcal{T}} \mathrm{RD}_{\text{dB}}(i,j) - \frac{1}{|\mathrm{RD}|}\sum_{i,j} \mathrm{RD}_{\text{dB}}(i,j),
	\label{eq:tloss}
\end{equation}
where $\mathcal{T}$ is the set of target bins, established using a prior estimate on the targets' position, followed by constant-false alarm rate detection and clustering, which is performed using the DBSCAN algorithm \cite{schubert2017dbscan}.

The waveform must be selected based on the history of waveforms and measurements which have been observed prior to the current time step. We denote the history, or \emph{information state} for step $t$ by $\mathcal{H}_{t} = \{w_{i},o_{i},\ell_{i}\}_{i=1}^{t}$.

However, at this point we encounter a practical challenge. Since the size of the information state grows with each time step, it becomes difficult to directly use the information state to select actions as $t \rightarrow \infty$. Thus, we may select actions based on a \emph{belief state} $b_{t}$, which is a set of parameters with fixed cardinality that closely approximates a sufficient statistic of the information state. The belief state is computed by the radar's tracking filter, and corresponds to the radar's current estimate of the target's 2d position and velocity. Ideally, $P(s_{t}|\mathcal{H}_{t}) \equiv P(s_{t}|b_{t})$.

From the radar's perspective, the uncertain dynamical model of the $m^{\text{th}}$ target can be represented as a Markov chain with transition probabilities
\begin{equation}
	p^{m}_{i,j} = P(s^{m}_{t} = j | s^{m}_{t-1} = i),
\end{equation}
which are unknown to the radar \emph{a priori}. We then apply the notation $P^{m}$ to denote the state-transition probability matrix for target $m$. We similarly describe the composite state-transition matrix of all the targets in the scene by $P$. Additionally, the uncertainty associated with observing a specific target measurement, given the target's state and a waveform, can be represented by the measurement likelihood function
\begin{equation}
	q^{m}_{i,j}(w_{t}) = P(o^{m}_{t} = j | s^{m}_{t} = i, w_{t}).
\end{equation}

The target motion and measurement models can be equivalently expressed in a functional form as
\begin{align}
	s_{t}^{m} &= f^{m}(s_{t-1}^{m}) + \nu^{m}_{t} \\
	o_{t}^{m} &= h^{m}(s_{t}^{m}) + \xi^{m}_{t}(w_{t}), 
\end{align}
where $f$ and $h$ are the target motion and measurement functions, which may be nonlinear in general, and $\nu^{i}$ and $\xi^{i}$ are noise terms which may be non-Gaussian in general. We note that the choice of waveform impacts the measurement process by through the covariance of the measurement noise $\xi^{m}_{t}$. For simple cases, such as single target tracking, while the target motion and measurement models are linear and the target can be reliably detected, optimal waveforms can be scheduled using the procedure described in \cite{kershaw1994optimal}. However, when such stringent assumptions are dropped, the waveform selection problem becomes more nuanced. Some works have considered numerical approximations to the basic approach \cite{kershaw1997waveform,Sira2007}. However, even for these approaches high SNR is assumed, and the computational requirements can become so burdensome as to inhibit the implementation of non-myopic optimization \cite[Chapter 6]{hero2007foundations}.

At this point we must consider important differences from the dynamic spectrum access scenario. First, in this setting, it less obvious how a rule-based waveform selection strategy should be selected, in comparison to the DSA setting, in which interference avoidance can be presumed to be the objective. Secondly, since there is a large space of possible target compositions, trajectories, and waveform catalogs, rule-based strategies that are designed under a particular set of assumptions are very likely to perform poorly in differing scenarios. Finally, in the multi-target tracking scenario, the number of states is expected to be very large, and developing waveform-selection rules for each possible state may be non-trivial.

In general, this problem is further complicated due to the following:
\begin{enumerate}
	\item Target motion and measurement models may in general be nonlinear, thus invalidating the Fisher information-based approach of \cite{kershaw1994optimal}.
	\item Discretization of state and observation spaces to an appropriate level is non-trivial without prior knowledge.
	\item Closed-form solutions are usually not available. 
\end{enumerate}

However, we note that the notions of entropy rate and diagonal dominance have similar interpretations in this scenario. Entropy rate generally corresponds to the \emph{maneuverability} of the targets while diagonal dominance corresponds to the \emph{stationarity} of the targets. We expect rule-based strategies to perform best when both entropy rate and diagonal dominance are small.

\subsection{Numerical Results}
In this section, we compare the performance of a TS-based waveform selection strategy to that of a fixed rule-based scheme. We primarily examine performance in terms of the loss function considered, under the assumption that improved range-Doppler SINR will result in favorable tracking performance. In the simulations, the rule-based waveform selection strategy is inspired by the simplified approach of \cite{rihaczek1971radar}, where the waveform is deterministically selected using the belief state $b_{k}$ and knowledge of the waveform ambiguity type. In \cite{rihaczek1971radar}, it is asserted that despite the infinite possibilities of radar waveforms, one may reasonably consider four broad classes of ambiguity functions, which depends on the time-bandwidth product, periodicity, and range-Doppler coupling. 

We consider a catalog of $|\ncalW| = 50$ waveforms. These include linear frequency modulated, exponential frequency modulated, Barker-13 phase-coded, Frank phase-coded, and Zadoff-Chu phase-coded waveforms. Within each of these classes, waveform parameters such as the bandwidth, FM sweep rate, number of sub-pulses, and sub-pulse length are varied.

Each radar CPI consists of 256 pulses, and the radar transmits a block of identical pulses during each CPI. The radar makes a decision to adapt its waveform is made on a CPI-to-CPI basis. The scene contains $M=3$ targets, each having a randomly generated state transition matrix. Constraints are placed on the number and locations of nonzero elements in the state transition matrix so that large changes in position and velocity are restricted. 

At the beginning of each track, the initial position, velocity, and trajectory of each of the three targets is randomized. The target frequency responses $h_{g}(f)$ and the clutter frequency responses $c_{h}(f)$ are Gaussian processes, and vary from CPI-to-CPI. The mean frequency response for each target is Gaussian shaped, and is expressed by
\begin{equation}
	\E[h_{g}(f)] = \exp(-(f/\beta_{g})^{2}),
\end{equation}
where $\beta_{g}$ is a target specific parameter that defines the spatial extent of the target. The mean clutter frequency responses are similarly defined by
\begin{equation}
	\E[c_{h}(f)] =  \exp(-(f/\gamma_{h})^{2}),
\end{equation}
where $\gamma_{h}$ determines the spatial extent of clutter component $h$.

\label{se:tracking}
\begin{figure*}
	\centering
	\includegraphics[scale=0.55]{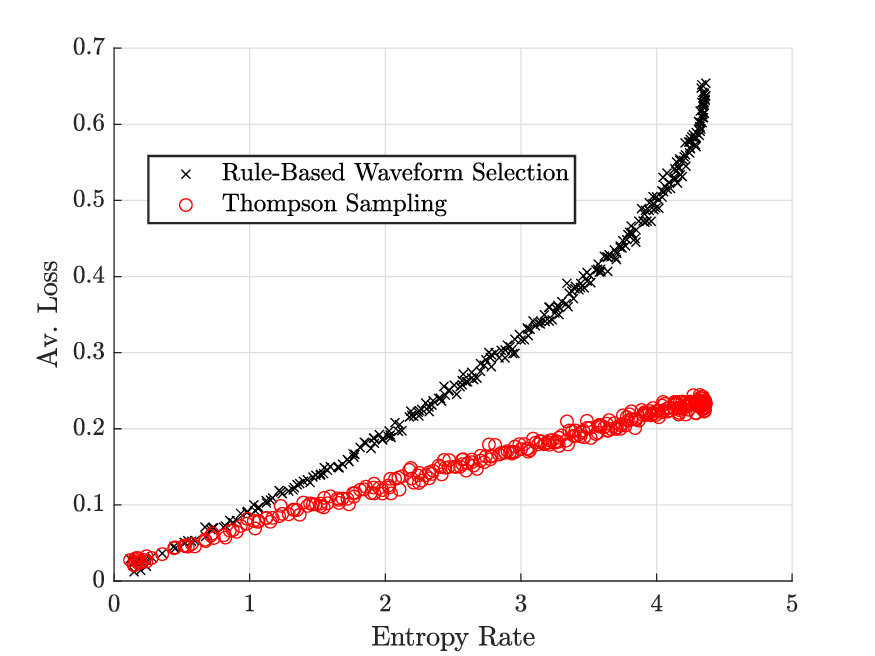}
	\includegraphics[scale=0.55]{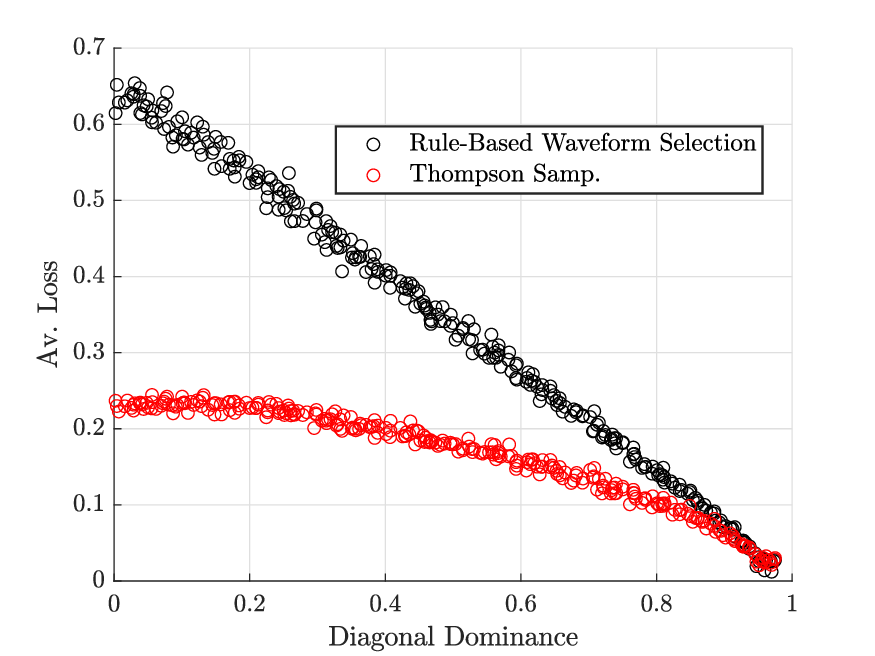}
	\caption{Variation in observed losses with entropy rate and diagonal dominance. We observe once again that for high entropy rates and low diagonal dominance values, TS reliably outperforms rule-based waveform selection even when the time-horizon is very limited.}
	\label{fig:track}
\end{figure*}

In Figure \ref{fig:track}, we observe the average loss per tracking instance, where each tracking instance lasts for $n=1e4$ radar CPIs. We examine performance for a variety of target motion and measurement models, resulting in a wide range of state transition matrices. We note that entropy rate is a strong inverse predictor of both TS and rule-based performance. Further, we note that diagonal dominance is a strong predictor of rule-based performance, although it is a weaker predictor of TS performance. This is due to the fact that diagonal dominance corresponds to the degree of target stationarity. For nearly stationary targets, using the most recent state estimate as a basis for waveform selection is expected to be effective. Additionally, entropy rate quantifies the `randomness' in the targets motion. Targets which quickly make sharp turns, or change velocity suddenly will result in a high entropy rate.

\begin{figure*}
	\centering
	\includegraphics[scale=0.5]{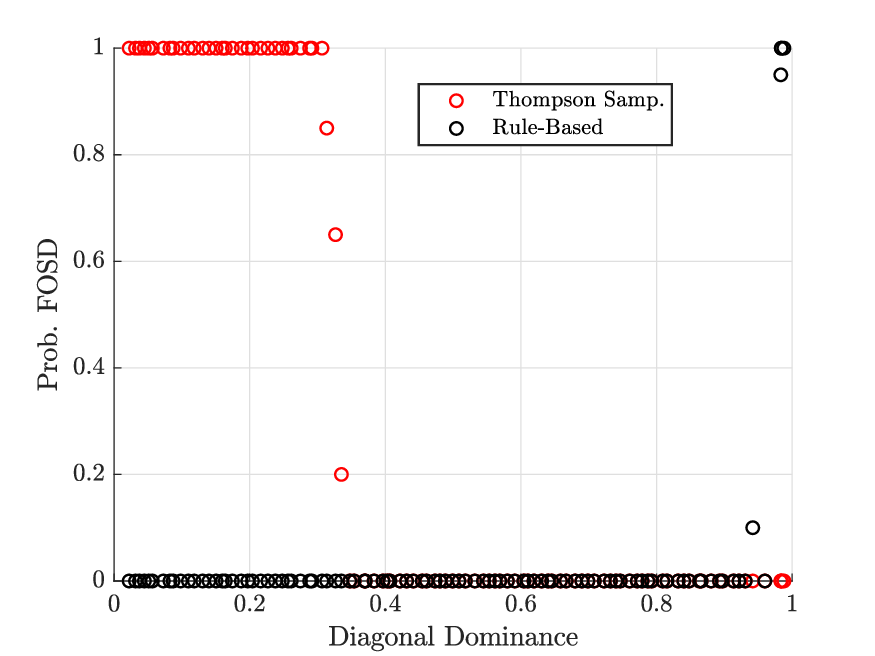}
	\includegraphics[scale=0.5]{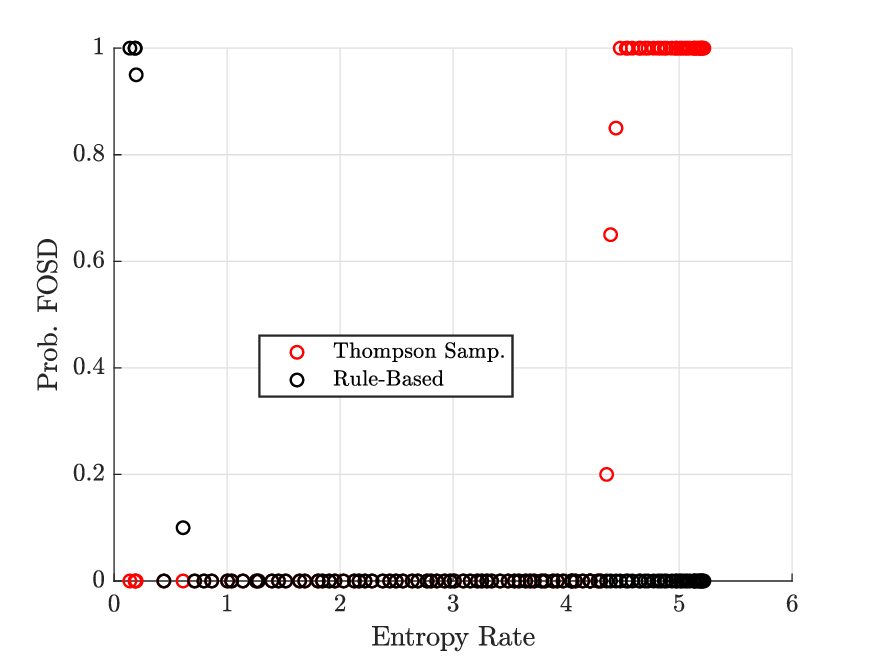}
	\caption{Diagonal dominance and entropy rate vs first order stochastic dominance.}
	\label{fig:fosdtrack}
\end{figure*}

In Figure \ref{fig:fosdtrack}, we observe the probability that the TS and rule-based strategies are first-order stochastically dominant with respect to the other strategy, over a range of entropy rate and diagonal dominance values. We test for stochastic dominance using a simple statistical test \cite{mcfadden1989testing}, having run each algorithm 100 times for $n=1e4$ at each value of diagonal dominance and entropy rate. We observe that when diagonal dominance is below $\approx .375$, TS is generally first-order dominant. Conversely, the rule-based strategy is only first order dominant for cases of diagonal dominance $\geq .925$. We observe that a similar trend holds with respect to entropy rate. For high entropy rates, TS is generally first-order dominant, and for entropy rates very close to zero, rule-based waveform selection is first-order dominant.

\begin{figure*}
	\centering
	\includegraphics[scale=0.5]{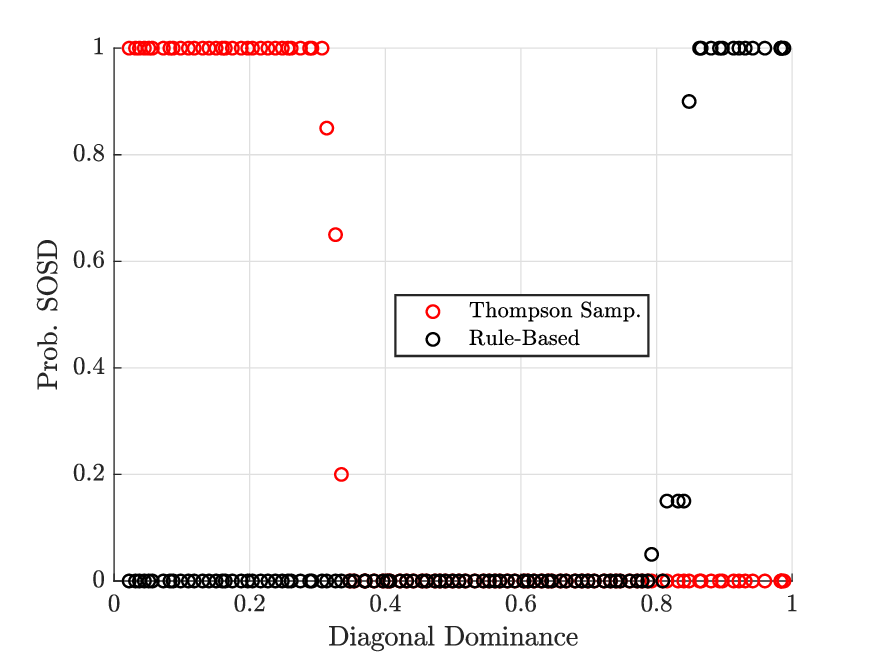}
	\includegraphics[scale=0.5]{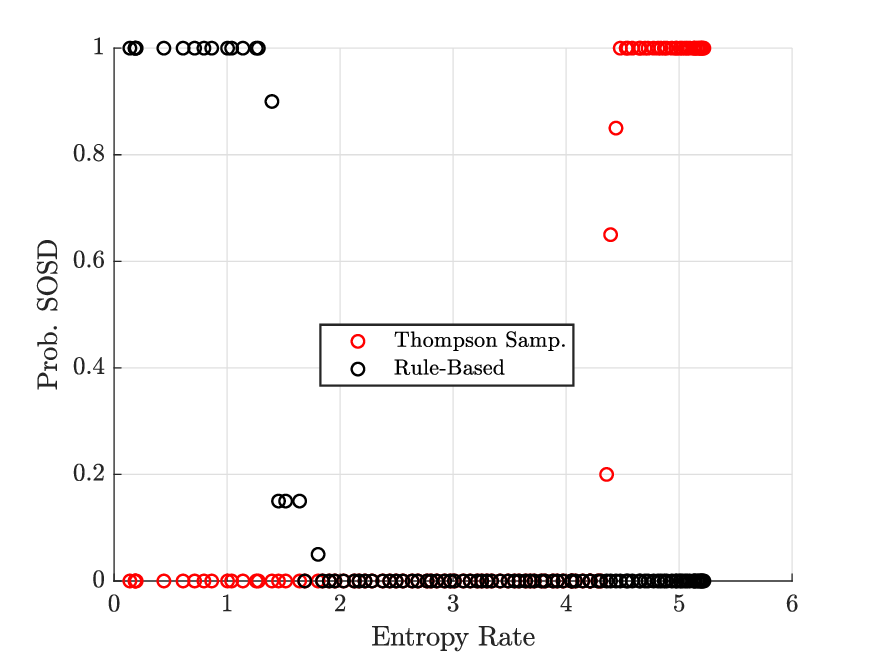}
	\caption{Diagonal dominance and entropy rate vs second order stochastic dominance.}
	\label{fig:sosdtrack}
\end{figure*}

Figure \ref{fig:sosdtrack} presents a similar results, but examines second-order stochastic dominance. We observe a similar trend as in the previous Figure, but note that the rule-based strategy begins to exhibit SOSD at diagonal dominance values as low as $\approx 0.78$. We note that while rule-based strategies are unlikely to perform strictly better than learning algorithms for cases of realistic target maneuverability, they may provide a simple and occasionally risk-averse solution.

\section{Conclusions and Insights for System Design}
\label{se:conclusions}
In this contribution, we have sought to answer a basic question about the value of learning-based cognitive radar by examining the performance of a Thompson Sampling-based waveform selection strategy as compared to a rule-based strategy in Markov interference and target tracking channels. We showed that the performance of any policy, or decision function, can be analyzed by decomposing the set of channel transitions into three subsets which characterize collision, missed opportunity, and successful transition events. We argue that a similar procedure can be generalized to compute the expected performance of policies in $K^{\text{th}}$ order Markov channels for various applications, as in \cite{thornton2022universal}. For stationary channel dynamics, the performance of any ML strategy can be decomposed into the performance during the exploration period and performance during exploitation.

Our numerical results demonstrate that given a time horizon of $10^{4}$ pulse repetition intervals, Thompson Sampling performs at least as well as the sense-and-avoid policy in terms of average SINR. However, for \emph{nearly stationary} interference channels over a more limited time horizon, sense-and-avoid outperforms TS by utilizing additional bandwidth while maintaining a nearly equivalent SINR. For Markov interference channels which are very dynamic, namely which have diagonal transition elements close to zero, TS provides large performance benefits, and generally dominates SAA in the first-order stochastic sense. Particularly, for channels which are very dynamic, but are close to deterministic, TS provides the largest performance benefits compared to SAA.

Online learning theory tells us that the most important factors which influence the regret of algorithms are the number of actions, required dimensionality of the context vector, and the time horizon itself. Thus, systems designers planning to employ learning-based strategies should account for the allowable time in which learning may occur, limit the waveform catalog to a manageable size, and represent the environment as efficiently as possible. We expect the general trend of the results shown here to hold for other learning algorithms, which optimize the loss functions defined in (\ref{eq:loss}) and (\ref{eq:tloss}). Some popular algorithms, such as those based on the upper-confidence bound principle, may provide favorable asymptotic performance as compared to TS, at the expense of sample-efficiency. Alternatively, for very simple waveform selection problems, optimal strategies based on the Gittins index may be feasible.

We also note that the observability of the scene's state plays a large role in the value of any deterministic waveform selection strategy. Machine learning algorithms are expected to provide large benefits over rule-based strategies if the observations are highly variable, which is the case when target measurements are incorporated into the model.

Future work will attempt to generalize this methodology to provide insights as to which environmental conditions make cognitive radar beneficial. For example, a more nuanced comparison involving more sophisticated fixed-rule based strategies would give greater insight into the practical value of cognitive radars. However, it is expected that next-generation cognitive radar systems will use a mixture of sequential decision making and rule-based strategies depending on the expected characteristics of the scene. Future work could consider using the definitions of stochastic dominance examined here to develop a principled algorithm selection scheme for a meta-cognitive radar.

\section*{Acknowledgments}
The authors would like to acknowledge and thank Daniel J. Jakubisin and Anthony F. Martone for several technical discussions which helped motivated and shape this work.

\appendices

\section{Proof of Theorem 2}
\begin{proof}[]
	Assuming perfect observability of the interference state (namely $o_{t} = s_{t} \quad \forall t$), the regret of $\psaa$ can be decomposed into two terms, the number of interference state transitions experienced, given by $R \in \nbbN$, and the sub-optimality gap $\Delta$ associated with each transition. Without loss of generality, let $R > 0$, $\Delta > 0$ be constants. Then it follows that SAA incurs regret linearly in $T$. The sub-linearity of TS regret has been well-established given a fixed number of arms $K$ and context dimension $d$, and is order $\tilde{\mathcal{O}}(d^{3/2}\sqrt{n})$. Thus, for any non-stationary Markov channel (namely, when some diagonal elements of $P$ are less than $1$), TS weakly dominates SAA in the asymptote. The extension to cases of non-perfect interference observability directly follow. However, in the finite horizon regime, precise specification of $K$, $T$, and $P$ are necessary to show stochastic dominance. It should be noted that for slowly varying channels, large values of $K$ and small $T$, SAA is more appealing.
\end{proof}

\section{Computation of Diagonal Dominance Measure}
\label{se:dominance}
Let $P$ be a $d \times d$ dimensional matrix, and let $j \in \mathbb{R}^{d}$ be a $d$-dimensional vector of all ones. Then let $r = [1,2,...,d]$ and $r_{2} = [1^{2},2^{2},...,d^{2}]$. We may compute the measure of diagonal dominance $r$ by first defining the following terms
\begin{align*}
        	n & =j A j^T \\
        	\Sigma x & =r A j^T \\
        	\Sigma y & =j A r^T \\
        	\Sigma x^2 & =r_2 A j^T \\
        	\Sigma y^2 & =j A r_2^T \\
        	\Sigma x y & =r A r^T
\end{align*} 
and applying the formula for the Pearson correlation coefficient
\begin{equation}
r=\frac{n \Sigma x y-\Sigma x \Sigma y}{\sqrt{n \Sigma x^2-(\Sigma x)^2} \sqrt{n \Sigma y^2-(\Sigma y)^2}}.
\end{equation}
We note that $r$ is bounded in $[-1,1]$, where $r=1$ corresponds to a matrix having only non-zero elements only on the main diagonal.

\begin{IEEEbiography}[{\includegraphics[width=1in,height=1.25in,clip,keepaspectratio]{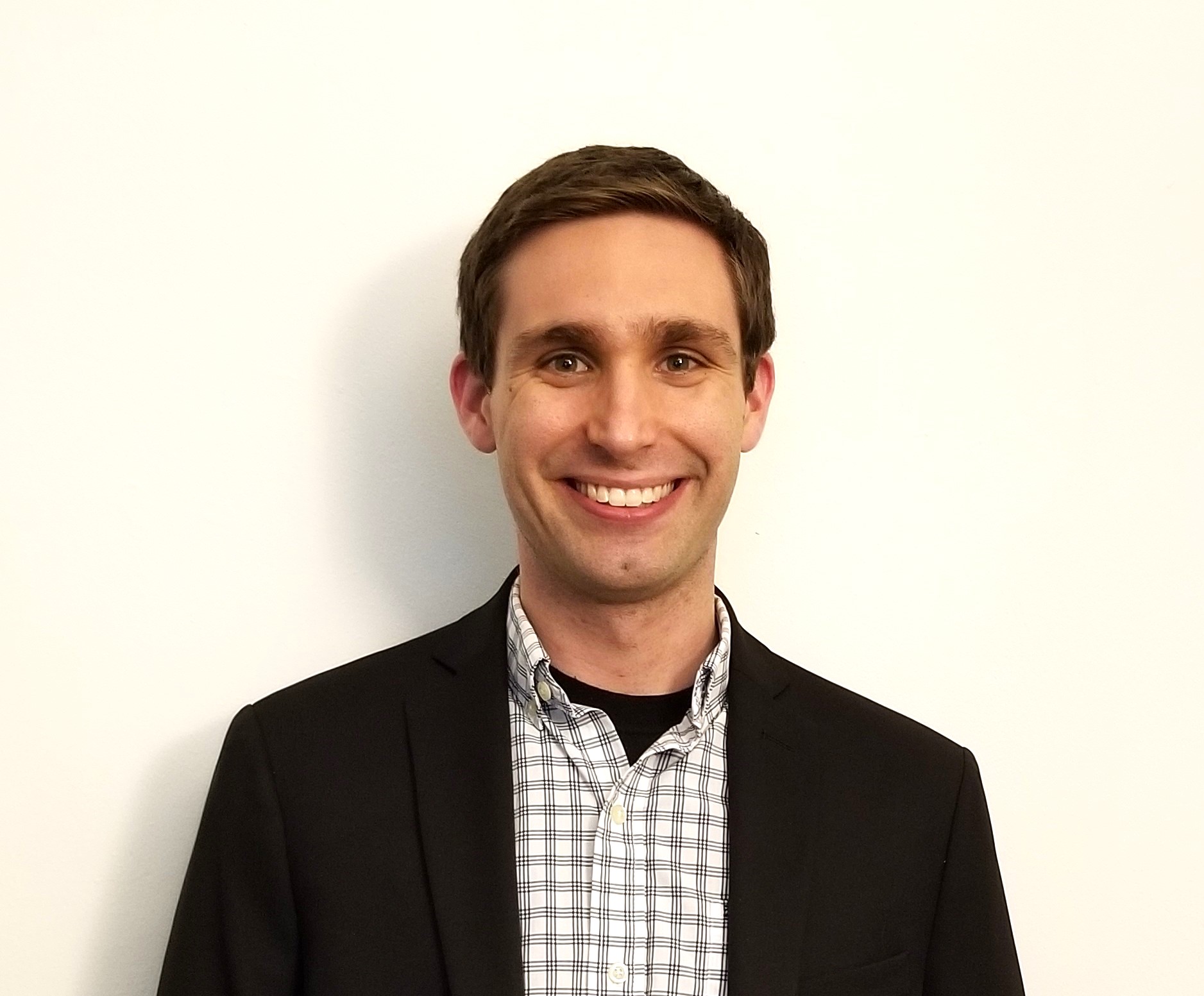}}]{Charles E. Thornton (M `23)}
	is a Research Assistant Professor with Virginia Tech's National Security Institute. He studied at Johns Hopkins University (BS, 2018) and Virginia Tech (PhD, 2023). His research interests include machine learning, information theory, and signal processing. He has made both technical and expository contributions to the literature on cognitive radar. 
\end{IEEEbiography}

\begin{IEEEbiography}[{\includegraphics[width=1in,height=1.25in,clip,keepaspectratio]{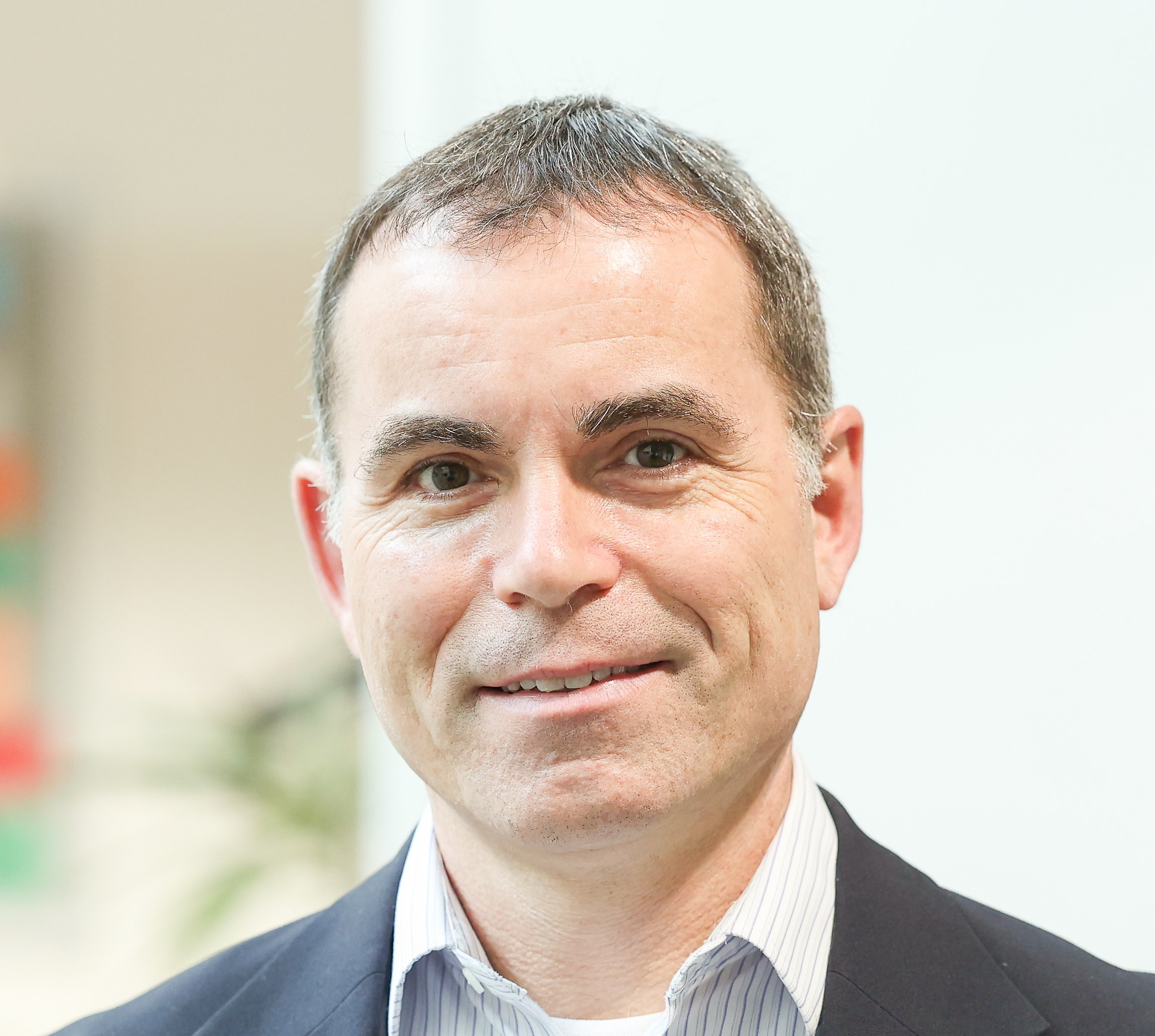}}]{R. Michael Buehrer (F `16)}
	joined Virginia Tech from Bell Labs as an Assistant Professor with the Bradley Department of Electrical and Computer Engineering in 2001. He is currently a Professor of Electrical Engineering and is the director of Wireless @ Virginia Tech, a comprehensive research group focusing on wireless communications, radar and localization. During 2009 Dr. Buehrer was a visiting researcher at the Laboratory for Telecommunication Sciences (LTS) a federal research lab which focuses on telecommunication challenges for national defense. While at LTS, his research focus was in the area of cognitive radio with a particular emphasis on statistical learning techniques. 
	
	Dr. Buehrer was named an IEEE Fellow in 2016 “for contributions to wideband signal processing in communications and geolocation.” His current research interests include machine learning for wireless communications and radar, geolocation, position location networks, cognitive radio, cognitive radar, electronic warfare, dynamic spectrum sharing, communication theory, Multiple Input Multiple Output (MIMO) communications, spread spectrum, interference avoidance, and propagation modeling. His work has been funded by the National Science Foundation, the Defense Advanced Research Projects Agency, the Office of Naval Research, the Army Research Office, the Air Force Research Lab and several industrial sponsors.
	
	Dr. Buehrer has authored or co-authored over 90 journal and approximately 260 conference papers and holds 18 patents in the area of wireless communications. In 2021 he was the co-recipient of the Vanu Bose Award for the best paper at MILCOM’21. In 2010 he was co-recipient of the Fred W. Ellersick MILCOM Award for the best paper in the unclassified technical program. He was formerly an Area Editor IEEE Wireless Communications. He was also formerly an associate editor for IEEE Transactions on Communications, IEEE Transactions on Vehicular Technologies, IEEE Transactions on Wireless Communications, IEEE Transactions on Signal Processing, IEEE Wireless Communications Letters, and IEEE Transactions on Education. He has also served as a guest editor for special issues of The Proceedings of the IEEE, and IEEE Transactions on Special Topics in Signal Processing. In 2003 he was named Outstanding New Assistant Professor by the Virginia Tech College of Engineering and in 2014 he received the Dean’s Award for Excellence in Teaching.
\end{IEEEbiography}

\bibliographystyle{IEEEtran}
\bibliography{valueBib}{}

% Generated by IEEEtran.bst, version: 1.14 (2015/08/26)
\begin{thebibliography}{10}
\providecommand{\url}[1]{#1}
\csname url@samestyle\endcsname
\providecommand{\newblock}{\relax}
\providecommand{\bibinfo}[2]{#2}
\providecommand{\BIBentrySTDinterwordspacing}{\spaceskip=0pt\relax}
\providecommand{\BIBentryALTinterwordstretchfactor}{4}
\providecommand{\BIBentryALTinterwordspacing}{\spaceskip=\fontdimen2\font plus
\BIBentryALTinterwordstretchfactor\fontdimen3\font minus
  \fontdimen4\font\relax}
\providecommand{\BIBforeignlanguage}[2]{{%
\expandafter\ifx\csname l@#1\endcsname\relax
\typeout{** WARNING: IEEEtran.bst: No hyphenation pattern has been}%
\typeout{** loaded for the language `#1'. Using the pattern for}%
\typeout{** the default language instead.}%
\else
\language=\csname l@#1\endcsname
\fi
#2}}
\providecommand{\BIBdecl}{\relax}
\BIBdecl

\bibitem{thornton2022cognitive}
C.~E. Thornton and R.~M. Buehrer, ``When is cognitive radar beneficial?
  insights from dynamic spectrum access,'' in \emph{IEEE Radar Conference},
  2023.

\bibitem{sira2009waveform}
S.~P. Sira, Y.~Li, A.~Papandreou-Suppappola, D.~Morrell, D.~Cochran, and
  M.~Rangaswamy, ``Waveform-agile sensing for tracking,'' \emph{IEEE Signal
  Processing Magazine}, vol.~26, no.~1, pp. 53--64, 2009.

\bibitem{calderbank2009waveform}
R.~Calderbank, S.~D. Howard, and B.~Moran, ``Waveform diversity in radar signal
  processing,'' \emph{IEEE Signal Processing Magazine}, vol.~26, no.~1, pp.
  32--41, 2009.

\bibitem{kershaw1994optimal}
D.~J. Kershaw and R.~J. Evans, ``Optimal waveform selection for tracking
  systems,'' \emph{IEEE Transactions on Information Theory}, vol.~40, no.~5,
  pp. 1536--1550, 1994.

\bibitem{gini2012waveform}
F.~Gini, A.~De~Maio, and L.~Patton, \emph{Waveform design and diversity for
  advanced radar systems}.\hskip 1em plus 0.5em minus 0.4em\relax Institution
  of engineering and technology London, UK, 2012.

\bibitem{blunt2016overview}
S.~D. Blunt and E.~L. Mokole, ``Overview of radar waveform diversity,''
  \emph{IEEE Aerospace and Electronic Systems Magazine}, vol.~31, no.~11, pp.
  2--42, 2016.

\bibitem{Sira2007}
S.~P. Sira, A.~Papandreou-Suppappola, and D.~Morrell, ``Dynamic configuration
  of time-varying waveforms for agile sensing and tracking in clutter,''
  \emph{IEEE Transactions on Signal Processing}, vol.~55, no.~7, pp.
  3207--3217, 2007.

\bibitem{martone2021closing}
A.~F. Martone, K.~D. Sherbondy, J.~A. Kovarskiy, B.~H. Kirk, R.~M. Narayanan,
  C.~E. Thornton, R.~M. Buehrer, J.~W. Owen, B.~Ravenscroft, S.~Blunt
  \emph{et~al.}, ``Closing the loop on cognitive radar for spectrum sharing,''
  \emph{IEEE Aerospace and Electronic Systems Magazine}, vol.~36, no.~9, pp.
  44--55, 2021.

\bibitem{kovarskiy2020evaluation}
J.~A. Kovarskiy, B.~H. Kirk, A.~F. Martone, R.~M. Narayanan, and K.~D.
  Sherbondy, ``Evaluation of real-time predictive spectrum sharing for
  cognitive radar,'' \emph{IEEE Transactions on Aerospace and Electronic
  Systems}, vol.~57, no.~1, pp. 690--705, 2020.

\bibitem{feng2018cognitive}
S.~Feng and S.~Haykin, ``Cognitive risk control for transmit-waveform selection
  in vehicular radar systems,'' \emph{IEEE Transactions on Vehicular
  Technology}, vol.~67, no.~10, pp. 9542--9556, 2018.

\bibitem{zamiri2022bayesian}
Y.~Zamiri-Jafarian and K.~N. Plataniotis, ``A bayesian surprise approach in
  designing cognitive radar for autonomous driving,'' \emph{Entropy}, vol.~24,
  no.~5, p. 672, 2022.

\bibitem{rihaczek1996principles}
A.~W. Rihaczek, ``Principles of high-resolution radar,'' \emph{Norwood, MA:
  Artech House, 1996.}, 1996.

\bibitem{bell1993information}
M.~R. Bell, ``Information theory and radar waveform design,'' \emph{IEEE
  Transactions on Information Theory}, vol.~39, no.~5, pp. 1578--1597, 1993.

\bibitem{Sowelam2000}
S.~Sowelam and A.~Tewfik, ``Waveform selection in radar target
  classification,'' \emph{IEEE Transactions on Information Theory}, vol.~46,
  no.~3, pp. 1014--1029, 2000.

\bibitem{wang2021deep}
C.~Wang, J.~Tian, J.~Cao, and X.~Wang, ``Deep learning-based uav detection in
  pulse-doppler radar,'' \emph{IEEE Transactions on Geoscience and Remote
  Sensing}, vol.~60, pp. 1--12, 2021.

\bibitem{kim2015firefighting}
J.-H. Kim, J.~W. Starr, and B.~Y. Lattimer, ``Firefighting robot stereo
  infrared vision and radar sensor fusion for imaging through smoke,''
  \emph{Fire Technology}, vol.~51, pp. 823--845, 2015.

\bibitem{ender2011radar}
J.~Ender, L.~Leushacke, A.~Brenner, and H.~Wilden, ``Radar techniques for space
  situational awareness,'' in \emph{2011 12th International Radar Symposium
  (IRS)}.\hskip 1em plus 0.5em minus 0.4em\relax IEEE, 2011, pp. 21--26.

\bibitem{hero2007foundations}
A.~O. Hero, D.~Casta{\~n}{\'o}n, D.~Cochran, and K.~Kastella, \emph{Foundations
  and applications of sensor management}.\hskip 1em plus 0.5em minus
  0.4em\relax Springer Science \& Business Media, 2007.

\bibitem{LaScala2004intl}
B.~La~Scala and B.~Moran, ``Measures of effectiveness for waveform selection,''
  in \emph{2004 International Waveform Diversity Design Conference}, 2004, pp.
  1--5.

\bibitem{mitchell2018cost}
A.~E. Mitchell, G.~E. Smith, K.~L. Bell, A.~J. Duly, and M.~Rangaswamy, ``Cost
  function design for the fully adaptive radar framework,'' \emph{IET Radar,
  Sonar \& Navigation}, vol.~12, no.~12, pp. 1380--1389, 2018.

\bibitem{charlish2020implementing}
A.~Charlish, K.~Bell, and C.~Kreucher, ``Implementing perception-action cycles
  using stochastic optimization,'' in \emph{2020 IEEE Radar Conference
  (RadarConf20)}.\hskip 1em plus 0.5em minus 0.4em\relax IEEE, 2020, pp. 1--6.

\bibitem{kershaw1997waveform}
D.~J. Kershaw and R.~J. Evans, ``Waveform selective probabilistic data
  association,'' \emph{IEEE Transactions on Aerospace and Electronic Systems},
  vol.~33, no.~4, pp. 1180--1188, 1997.

\bibitem{nguyen2015adaptive}
N.~H. Nguyen, K.~Dogancay, and L.~M. Davis, ``Adaptive waveform selection for
  multistatic target tracking,'' \emph{IEEE Transactions on Aerospace and
  Electronic Systems}, vol.~51, no.~1, pp. 688--701, 2015.

\bibitem{rihaczek1971radar}
A.~W. Rihaczek, ``Radar waveform selection-a simplified approach,'' \emph{IEEE
  Transactions on aerospace and electronic systems}, no.~6, pp. 1078--1086,
  1971.

\bibitem{charlish2015anticipation}
A.~Charlish and F.~Hoffmann, ``Anticipation in cognitive radar using stochastic
  control,'' in \emph{2015 IEEE Radar Conference (RadarCon)}.\hskip 1em plus
  0.5em minus 0.4em\relax IEEE, 2015, pp. 1692--1697.

\bibitem{selvi2020reinforcement}
E.~Selvi, R.~M. Buehrer, A.~Martone, and K.~Sherbondy, ``Reinforcement learning
  for adaptable bandwidth tracking radars,'' \emph{IEEE Transactions on
  Aerospace and Electronic Systems}, vol.~56, no.~5, pp. 3904--3921, 2020.

\bibitem{thornton2020deep}
C.~E. Thornton, M.~A. Kozy, R.~M. Buehrer, A.~F. Martone, and K.~D. Sherbondy,
  ``Deep reinforcement learning control for radar detection and tracking in
  congested spectral environments,'' \emph{IEEE Transactions on Cognitive
  Communications and Networking}, vol.~6, no.~4, pp. 1335--1349, 2020.

\bibitem{ahmad2009optimality}
S.~H.~A. Ahmad, M.~Liu, T.~Javidi, Q.~Zhao, and B.~Krishnamachari, ``Optimality
  of myopic sensing in multichannel opportunistic access,'' \emph{IEEE
  Transactions on Information Theory}, vol.~55, no.~9, pp. 4040--4050, 2009.

\bibitem{krishnamurthy2009optimal}
V.~Krishnamurthy and D.~V. Djonin, ``Optimal threshold policies for
  multivariate pomdps in radar resource management,'' \emph{IEEE transactions
  on Signal Processing}, vol.~57, no.~10, pp. 3954--3969, 2009.

\bibitem{krishnamurthy2016partially}
V.~Krishnamurthy, \emph{Partially observed Markov decision processes}.\hskip
  1em plus 0.5em minus 0.4em\relax Cambridge university press, 2016.

\bibitem{krishnamurthy2002algorithms}
------, ``Algorithms for optimal scheduling and management of hidden markov
  model sensors,'' \emph{IEEE Transactions on Signal Processing}, vol.~50,
  no.~6, pp. 1382--1397, 2002.

\bibitem{thornton2021constrained}
C.~E. Thornton, R.~M. Buehrer, and A.~F. Martone, ``Constrained contextual
  bandit learning for adaptive radar waveform selection,'' \emph{IEEE
  Transactions on Aerospace and Electronic Systems}, vol.~58, no.~2, pp.
  1133--1148, 2021.

\bibitem{thornton2022online}
------, ``Online bayesian meta-learning for cognitive tracking radar,''
  \emph{arXiv preprint arXiv:2207.06917}, 2022.

\bibitem{howard2022distributed}
W.~W. Howard, A.~F. Martone, and R.~M. Buehrer, ``Distributed online learning
  for coexistence in cognitive radar networks,'' \emph{IEEE Transactions on
  Aerospace and Electronic Systems}, 2022.

\bibitem{griffiths2014radar}
H.~Griffiths, L.~Cohen, S.~Watts, E.~Mokole, C.~Baker, M.~Wicks, and S.~Blunt,
  ``Radar spectrum engineering and management: Technical and regulatory
  issues,'' \emph{Proceedings of the IEEE}, vol. 103, no.~1, pp. 85--102, 2014.

\bibitem{Kirk2019}
B.~H. Kirk, R.~M. Narayanan, K.~A. Gallagher, A.~F. Martone, and K.~D.
  Sherbondy, ``Avoidance of time-varying radio frequency interference with
  software-defined cognitive radar,'' \emph{IEEE Transactions on Aerospace and
  Electronic Systems}, vol.~55, no.~3, pp. 1090--1107, 2019.

\bibitem{sira2006waveform}
S.~P. Sira, A.~Papandreou-Suppappola, D.~Morrell, and D.~Cochran,
  ``Waveform-agile sensing for tracking multiple targets in clutter,'' in
  \emph{2006 40th Annual Conference on Information Sciences and Systems}.\hskip
  1em plus 0.5em minus 0.4em\relax IEEE, 2006, pp. 1418--1423.

\bibitem{Shi2022}
C.~Shi, Y.~Wang, S.~Salous, J.~Zhou, and J.~Yan, ``Joint transmit resource
  management and waveform selection strategy for target tracking in distributed
  phased array radar network,'' \emph{IEEE Transactions on Aerospace and
  Electronic Systems}, vol.~58, no.~4, pp. 2762--2778, 2022.

\bibitem{heger1994consideration}
M.~Heger, ``Consideration of risk in reinforcement learning,'' in \emph{Machine
  Learning Proceedings 1994}.\hskip 1em plus 0.5em minus 0.4em\relax Elsevier,
  1994, pp. 105--111.

\bibitem{keramati2020being}
R.~Keramati, C.~Dann, A.~Tamkin, and E.~Brunskill, ``Being optimistic to be
  conservative: Quickly learning a cvar policy,'' in \emph{Proceedings of the
  AAAI conference on artificial intelligence}, vol.~34, no.~04, 2020, pp.
  4436--4443.

\bibitem{thornton2022universal}
C.~E. Thornton, R.~M. Buehrer, H.~S. Dhillon, and A.~F. Martone, ``Universal
  learning waveform selection strategies for adaptive target tracking,''
  \emph{IEEE Transactions on Aerospace and Electronic Systems}, 2022.

\bibitem{kirk2020performance}
B.~H. Kirk, A.~F. Martone, K.~D. Sherbondy, and R.~M. Narayanan, ``Performance
  analysis of pulse-agile sdradar with hardware accelerated processing,'' in
  \emph{2020 IEEE International Radar Conference (RADAR)}.\hskip 1em plus 0.5em
  minus 0.4em\relax IEEE, 2020, pp. 117--122.

\bibitem{lattimore2020bandit}
T.~Lattimore and C.~Szepesv{\'a}ri, \emph{Bandit algorithms}.\hskip 1em plus
  0.5em minus 0.4em\relax Cambridge University Press, 2020.

\bibitem{agrawal2013thompson}
S.~Agrawal and N.~Goyal, ``Thompson sampling for contextual bandits with linear
  payoffs,'' in \emph{International conference on machine learning}.\hskip 1em
  plus 0.5em minus 0.4em\relax PMLR, 2013, pp. 127--135.

\bibitem{Hochwald1999}
B.~Hochwald and P.~Jelenkovic, ``State learning and mixing in entropy of hidden
  markov processes and the gilbert-elliott channel,'' \emph{IEEE Transactions
  on Information Theory}, vol.~45, no.~1, pp. 128--138, 1999.

\bibitem{Kirk2023}
B.~H. Kirk, A.~F. Martone, K.~A. Gallagher, R.~M. Narayanan, and K.~D.
  Sherbondy, ``Mitigation of clutter modulation in cognitive radar for spectrum
  sharing applications,'' \emph{IEEE Transactions on Radar Systems}, vol.~1,
  pp. 193--204, 2023.

\bibitem{schubert2017dbscan}
E.~Schubert, J.~Sander, M.~Ester, H.~P. Kriegel, and X.~Xu, ``Dbscan revisited,
  revisited: why and how you should (still) use dbscan,'' \emph{ACM
  Transactions on Database Systems (TODS)}, vol.~42, no.~3, pp. 1--21, 2017.

\bibitem{mcfadden1989testing}
D.~McFadden, ``Testing for stochastic dominance,'' in \emph{Studies in the
  economics of uncertainty: In honor of Josef Hadar}.\hskip 1em plus 0.5em
  minus 0.4em\relax Springer, 1989, pp. 113--134.

\end{thebibliography}

\end{document}